\documentclass[11pt,a4paper,titlepage]{article}
\usepackage[utf8]{inputenc}
\usepackage[T1]{fontenc}
\usepackage[english]{babel}
\usepackage{mathrsfs}
\usepackage{amsmath}
\usepackage{amsfonts}
\usepackage{amssymb}
\usepackage{amsthm}
\usepackage{longtable}
\usepackage{enumerate}
\usepackage{bussproofs}
\usepackage{proof}
\usepackage{xcolor}
\usepackage[width=0.00cm, height=2.50cm, left=2.50cm, right=3.00cm, top=2.50cm, bottom=2.50cm]{geometry}
\usepackage{url}

\theoremstyle{definition}
\newtheorem{definition}{Definition}[section]
\newtheorem{theorem}{\textbf{Theorem}}
\newtheorem{lemma}{\textbf{Lemma}}
\newtheorem{prop}{ \textbf{Proposition}}

\newtheorem{hypo}{\textbf{Hypothesis}}
\newtheorem{remark}{\textbf{Remark}}

\newcommand{\F}{\mathscr{L}}				

\newcommand{\Lang}{\mathscr{L}}	
\newcommand{\Prop}{\mathcal{P}}			
\newcommand{\B}{\Box}	
\newcommand{\D}{\Diamond}
\newcommand{\var}{\vartheta}
\newcommand{\G}{\Gamma}

\newcommand{\non}{{\sim}}

\newcommand{\FE}{\widetilde{\circ}}
\newcommand{\FA}{\widetilde{\bullet}}
\newcommand{\one}{T}
\newcommand{\zero}{F}
 \newcommand{\A}{A}
 \newcommand{\Be}{B}
 \newcommand{\Ce}{C}
\newcommand{\ts}{\,}
  \newcommand{\f}{\mathfrak{f}}
  \newcommand{\g}{\mathfrak{g}}

\usepackage{bussproofs,longtable,multicol}
\usepackage[breaklinks,colorlinks,citecolor=blue,urlcolor=blue]{hyperref}
\newcommand\DOI[2]{DOI: \href{#1#2}{#2}}

\newenvironment{Published}{\begin{list}{}{\leftmargin=25pt}
    \item[]\small\noindent{Published as}}%
    {\end{list}}
    \newenvironment{Funding}{\begin{list}{}{\leftmargin=25pt}
        \item[]\small\noindent{The research}}%
        {\end{list}}
\newenvironment{Abstract}{\begin{list}{}{\leftmargin=25pt}
    \item[]\footnotesize\noindent{\sc \bf Abstract:}}%
{\end{list}}
\newenvironment{keywords}{\begin{list}{}{\leftmargin=25pt}
    \item[]\footnotesize\noindent{\sc \bf Keywords:}}%
{\end{list}}

\begin{document}
\begin{center}
\vspace{0.3cm}
{\large  
Yaroslav Petrukhin\vspace{0.3cm}\\
\textbf{Essence and accident modalities meet Belnapian truth values}}

\ \\
\begin{footnotesize}
 Center for Philosophy of Nature, University of Lodz, Poland,  \url{yaroslav.petrukhin@gmail.com}\\
\end{footnotesize}
\end{center}
\begin{Published}
Petrukhin, Y. (2026)
Essence and accident modalities meet Belnapian truth values. Studia Logica. Online first paper. 
\DOI{http://dx.doi.org/}{10.1007/s11225-026-10250-z}
\bigskip
\end{Published}
\begin{Funding}
 was funded by the European Union (ERC, ExtenDD, project number: 101054714). Views and opinions expressed are however those of the author(s) only and do not necessarily reflect those of the European Union or the European Research Council. Neither the European Union nor the granting authority can be held responsible for them. 
\end{Funding}
 \begin{Abstract}
This paper investigates many-valued generalisations of the classical
essence and accident modalities. In two-valued logic, a proposition is
essentially true (resp.\ false) if, whenever it is true (resp.\ false), it is necessarily true (resp.\ false); it is accidentally true (resp.\ false) if it is true (resp.\ false) but not necessarily so. Many-valued logics provide a natural setting for introducing further modalities of this kind. We focus on Belnap--Dunn's First-Degree Entailment
(\textbf{FDE}), a four-valued system that generalises the classical truth values. More precisely, we consider an extension of \textbf{FDE} with Boolean negation and implication. In addition to modalities of
essential and accidental truth and falsity, we define modalities of essential and accidental inconsistency and indeterminacy. We present a four-valued \textbf{S5}-based Kripke semantics and cut-free hypersequent
calculi for the resulting logics. We then prove semantic and syntactic embedding theorems for these logics into a four-valued version of \textbf{S5} with necessity and possibility modalities. These embeddings clarify the intended interpretation of the Belnapian essence and accident modalities and yield soundness, completeness, and cut-admissibility results.
 \end{Abstract}
 \begin{keywords}
Hypersequent calculus; cut elimination; proof theory; modal logic; many-valued logic; Belnap logic; essence logic; accident logic.
 \end{keywords}

\section{Introduction}
Marcos~\cite{Marcos} introduced the modalities of essence and accident in terms of the standard necessity and possibility operators and Boolean connectives as follows: $\circ A=A\rightarrow\Box A$ and $\bullet A=\non {\circ} A=A\wedge\non\Box A=A\wedge\D\non A$. As he notes in another paper, `one could read $ \bullet A$ as saying that `$ A $ is the case, but could have been otherwise': It works as a kind of (local) connective for `accidental truth'. Similarly, $ \circ $ could be read as expressing a (local) notion of `essential truth''~\cite[p.~297, notation adjusted]{MarcosPara}. This interpretation is also adopted by Steinsvold~\cite{Steinsvold,Steinsvoldboxdot}, and we will follow it here. Thus, we read $\circ A$ as `$A$ is essentially true', that is, `if $A$ is true, then $A$ is necessarily true', and we read $\bullet A$ as `$A$ is accidentally true', that is, `$A$ is true but possibly false'.

One can also consider two further modalities: $ \FE A=\non A\rightarrow \Box \non A$ and $ \FA A=\non A\wedge \D A$.  Marcos~\cite{MarcosPara} interprets them as connectives of undeterminedness and determinedness, respectively. `One could read $\FA A $ as saying that `$ A $ is not the case, but it could have been'. This suggests that $\FA $ could work as a kind of (local) connective for `counterfactual truth''~\cite[p.~297, notation adjusted]{MarcosPara}. In~\cite{PetrukhinS5}, $\FE A$ is interpreted as `$A$ is essentially
false', that is, `if $A$ is false, then $A$ is necessarily false', and $\FA A$ is interpreted as `$A$ is accidentally false', that is, `$A$ is false but $A$ is possibly true'. We will follow this latter interpretation.

Gilbert and Venturi~\cite{GilbertVenturi} criticise the Marcos--Steinsvold's interpretation of essence and accident. In particular, they point out that in the logic of essence based on the modal logic $\mathbf{K}$ the formula $\circ \bot$ is valid, which means that $\bot$ is essentially true. This sounds rather counter-intuitive and motivates Gilbert and Venturi to introduce their own operators for essence and accident: $\mathsf{A}A :=\bullet A \vee\bullet\non A$ and $\mathsf{E}A=\non \mathsf{A}A$. These formulas are supposed to be read as `$A$ is accidental' and `$A$ is essential', respectively. `As for `accidentally true' and `essentially true', these can now be given straightforward formalizations as $A \wedge \mathsf{A}A$ and $A \wedge \mathsf{E}A$, respectively. (Note that $A \wedge \mathsf{A}A$ is equivalent to just $\bullet A$, as desired). Additionally, and interestingly, it is worth noting that while $\mathsf{E}A$ is not equivalent to $\circ A$, it is the case that $A \wedge \mathsf{E}A$ is equivalent to $A \wedge \circ A$. This nicely mimics the fact that while $A$ and $\circ A$ are not equivalent, $A\wedge\Box A$ is equivalent to $A\wedge\circ A$.' 
\cite[p.~890]{GilbertVenturi}. Thus, Gilbert and Venturi's interpretation of `accidentally true' coincides with that of Marcos, whereas their interpretation of `essentially true' is different.

Let us write $\boxdot A$ for $A \wedge \circ A = A \wedge \Box A$. We choose this symbol because it was used by Boolos~\cite{Boolos}, who studied this modality in the context of provability logic and interpreted it as `true and provable (in Peano arithmetic)'. It is easy to observe that in reflexive modal logics $A \wedge \Box A$ is equivalent to $\Box A$, thanks to the axiom $\Box A \rightarrow A$. One could similarly define `accidental falsity' as $\neg A \wedge \Box \neg A$, but this is equivalent to $\Box \neg A$. Hence, if we want to study $\boxdot A$ as an `essentially true' modality and do not wish it to collapse to $\Box A$, we have to work in a non-reflexive modal logic; this illustrates a limitation of this approach. Nevertheless, such a study could be carried out and might produce fruitful results.

In the present paper, however, we prefer to follow the Marcos--Steinsvold interpretation of essence and accident. We think that their understanding of `essentially true' is best viewed as a notion of \emph{conditionally} essentially true. Indeed, if we say that $\circ \bot$ is conditionally essentially true, we do not obtain the counter-intuitive claim that `$\bot$ is essentially true', but only the statement: `if $\bot$ were true, then it would be true at all accessible worlds'.

If we switch from classical logic to some many-valued one, a natural generalization of these modalities arises: $ A $ essentially takes the value $ i $ and $ A $ accidentally takes the value $ i $, $ \circ_iA= A^i\rightarrow \Box A^i$ and $ \bullet_iA= A^i\wedge \D\non A^i$, where $ A^i $ is a formula meaning that $ A $ takes the value $ i $. Here and after $ \non $ is Boolean negation and $ \neg $ is De Morgan one. A negative generalisation is possible as well: $ A $ essentially does not take the value $ i $ and $ A $ accidentally does not take the value $ i $, $ \FE_iA $ and $ \FA_iA $, respectively. If the value $ i $ does not coincide with any classical truth value, we get a modality that cannot be formulated in the framework of two-valued modal logic and can be expressed only in the framework of many-valued modal logic. In our case, $ i\in \{t,b,n,f\} $, since we consider a logic with the same truth values as Belnap--Dunn's \textbf{FDE} \cite{Belnap1,Belnap2,Dunn76}. The set $ \{t,b,n,f \} $ of truth values of \textbf{FDE} generalises the set $ \{\one,\zero\} $ of classical truth values:  $ t $ (`true') stands for $ \{\one\} $, $ b $ (`both true and false') for $ \{\one,\zero\} $, $ n $ (`neither true nor false') for $ \emptyset $, and $ f $ (`false') for $ \{\zero\} $.\footnote{The nature of \textbf{FDE}'s values is epistemic: a formula $ A $ takes the value $ t $, if Belnap's computer is told that $ A $ is true, $ A $ takes the value $ f $, if Belnap's computer is told that $ A $ is false, $ A $ takes the value $ b $, if Belnap's computer gets contradictory information, and $ A $ takes the value $ n $, if Belnap's computer has no information. Hence, the nature of essence and accident modalities in a \textbf{FDE}-style logic is epistemic too.} The values $ t $ and $ b $ are designated. The modalities $\circ_b$ and $\bullet_b$ correspond to essential and
accidental inconsistency, and $\circ_n$ and $\bullet_n$ to essential and accidental indeterminacy.
  
More precisely, our research focuses on an expansion of \textbf{FDE} incorporating Boolean negation and implication. It is worth noting that \textbf{FDE} already has a De Morgan negation that is both paraconsistent and paracomplete. Our logic can be described either as Avron's \cite{Avron} logic \textbf{HBe} (also known as Popov's logic \textbf{Par} \cite{Popov}) extended by Boolean negation, or as De and Omori's logic \textbf{BD+} \cite{DeOmori} in which the implication connective is replaced by that of \textbf{HBe}. The matrix definition of the implication of \textbf{HBe} is given in \cite{ArieliAvron}.\footnote{We find the behaviour of De Morgan negation of the implication of \textbf{HBe} more natural, $ \neg(A\rightarrow B)=A\wedge \neg B$ than the behaviour of De Morgan negation of the implication of \textbf{BD+}, $ \neg(A\rightarrow B)=\non\neg A\wedge \neg B$.} The presence of two negations in a language is advantageous for providing precise and natural definitions for many-valued essence and accident (see Definitions \ref{i} and \ref{g}). Given that one of the negations is both paraconsistent and paracomplete, it is still possible to handle contradictory and partial information to some extent. Implication is primarily required for defining the essence modality. 
  
In addition to $ \circ_i $ and $ \bullet_i $, where $ i\in\{t,b,n,f\} $, we introduce another many-valued version of essence and accident: $ \circ_I $ and $ \bullet_I $, where $ I\in\{T,F \} $. Following Belnap \cite{Belnap1,Belnap2}, we say that a formula is at least true, if $ v(A)\in\{t,b\} $, that is, $ T\in v(A)$, where $ v $ is a valuation; and a formula is at least false, if $ v(A)\in\{f,b\} $, that is, $ F\in v(A)$. Thus, $ \circ_TA $ (resp. $ \bullet_TA $) means that $ A $ is essentially (resp. accidentally) at least true (or essentially (resp. accidentally) takes the designated values), and $ \circ_FA $ (resp. $ \bullet_FA $) means that $ A $ is essentially (resp. accidentally) at least false. We also define negative versions of these modalities: $ \FE_I $ and $ \FA_I $, where $ I\in\{T,F \} $. $ \FE_TA $ (resp. $ \FA_TA $) means that $ A $ essentially (resp. accidentally) is not at least true, and $ \FE_FA $ (resp. $ \FA_FA $) means that $ A $ essentially (resp. accidentally) is not at least false. In subsequent sections, we give more formal definitions of these modalities. 
 
Belnap's truth values generalise classical truth values in a natural way, hence, an \textbf{FDE}-style logic accommodates a many-valued generalization of essence and accident modalities, enabling the differentiation of numerous dimensions of truth and falsity (e.g., being true and being at least true). This renders a logic based on Belnap's truth values an excellent foundation for the exploration of various modalities. Observe that numerous works focus on modal logics based on \textbf{FDE} itself and its extensions: as an example, let us mention the papers by Drobyshevich \cite{Drobyshevich}, Fitting \cite{Fitting91,Fitting92}, Goble \cite{Goble}, Priest \cite{Priest08b}, Odintsov and Wansing \cite{OdintsovWansing17,OdintsovWansing}, Odintsov, Skurt, and Wansing \cite{OdintsovSkurtWansing}, 
Odintsov and Speranski \cite{OdintsovSperanski,OdintsovSperanski19}, and  Rivieccio, Jung, and Jansana \cite{Rivieccio4}. 
 These papers are devoted to the investigation of necessity and possibility. Kozhemiachenko and Vashentseva  explored a non-contingency modality \cite{KozhVashentsevaNonCont} as well as the notion of ignorance \cite{KozhVashentsevaIgn} in \textbf{FDE}. Sedl\'{a}r \cite{Sedlar} investigated a version of propositional dynamic logic having Belnapian truth values. To the best of our knowledge, the only existing accident operator for \textbf{FDE} is proposed in \cite{KozhVashentsevaIgn}; unlike that operator, the essence and accident modalities studied in the present paper are definable via the necessity operator. Just as the works mentioned above develop four-valued \textbf{FDE}-based counterparts of standard two-valued modal logics with necessity and possibility,  dynamic logic, logics of non-contingency and ignorance, we propose four-valued analogues of logics of essence and accident. In the Belnapian setting, this also allows us to distinguish new ‘dimensions' of essence and accident, beyond the classical true and false case.
 
 We use a four-valued version of \textbf{S5}, following Priest's \cite{Priest08b} definition of necessity and possibility modalities in the setting of Belnapian truth values.  
 We want to not only semantically describe many-valued essence and accident modalities, but also develop proof theory for them. Proof theory is well-developed for all main modal logics, but \textbf{S5} stands out as a particularly straightforward and helpful tool for investigating classical essence and accident modalities. This paper is a continuation of \cite{PetrukhinS5}, where Restall's \cite{Restall} hypersequent calculus for \textbf{S5} (originally, having rules only for one modality, that is $ \Box $) has been extended by the rules for $ \circ $, $ \bullet $, $ \FE $, and $ \FA $. Following \cite{PetrukhinS5}, we choose Restall's \cite{Restall} hypersequent calculus for \textbf{S5} which is one of the simplest hypersequent calculi for this logic, having no restrictions on the context formulas in the modal rules, which simplifies their formulation. Although ordinary sequent calculi for \textbf{S5} are not cut-free, there are several cut-free hypersequent calculi with excellent proof-theoretic behaviour (see \cite{BednarskaIndrzejczak} for a survey) and several other generalisations of sequent calculi as well (see \cite{Poggiolesi} for a survey). The absence of a cut-free ordinary sequent calculus for \textbf{S5} has motivated proof theorists to create various more comprehensive types of Gentzen-style calculi for \textbf{S5}. As a result, this logic has become one of the most extensively studied and appealing from a proof-theoretical perspective.
  
 The structure of this paper is as follows. Section \ref{Semantics} is devoted to the formulation of the semantics. Section \ref{Hypersequent} presents a hypersequent calculus. Section \ref{Embedding} contains the formulation of the embedding functions, syntactic and semantic embedding procedures of two types (first, we embed the four-valued version of \textbf{S5} with $ \Box $ and $ \D $ into \textbf{S5} itself, then we embed the four-valued version of \textbf{S5} with essence and accident into the four-valued version of \textbf{S5} with $ \Box $ and $ \D $). As a result of the embedding theorems, we obtain soundness, completeness, cut admissibility, and decidability. Section \ref{Conclusion} contains concluding comments.
 
 \section{Semantics}\label{Semantics}
 Let $ \Prop $ be a set of propositional variables $ \{p_1,p_2,\ldots \} $ and $ \Prop^\prime $ be an additional set of propositional variables $ \{p^\prime \mid p\in\Prop \} $ needed for simplifying embedding procedures (see, e.g. \cite{Kamide19,Kamide23,MML,MMLS5}). Let us consider the language $ \Lang_{\B\D} $ with the alphabet $ \langle \Prop,\Prop^\prime,\non,\B,\D,\wedge,\vee,\rightarrow,(,) \rangle $. Let $ \Lang_{\B\D}^\neg $ be the language obtained from $ \Lang_{\B\D} $ by deleting $ \Prop^\prime $ and adding one more negation $ \neg$. Let $\Lang_{\circ\bullet}^\neg $ be the language obtained from $ \Lang_{\B\D}^\neg $ by the replacement of $ \B $ and $ \D $ with at least one of the following modal operators: $ \circ_I $, $ \bullet_I $, $ \FE_I $, $ \FA_I $, $ \circ_i $, $ \bullet_i $, $ \FE_i $, and $ \FA_i $, where $ I\in\{T,F\} $ and $ i\in\{t,b,n,f \}  $.  We define $\Lang_{\B\D} $-, $ \Lang_{\B\D}^\neg $-, and $ \Lang_{\circ\bullet}^\neg $-formulas in the standard inductive way.
 
 \begin{definition}\label{Kripkeval}
 A pair $ \langle W,\var\rangle $ is said to be an \textbf{S5}-model iff $ W\not=\emptyset $ and $ \var $ is a mapping from $ W\times\F_{\Box\D} $ to $\{\one,\zero\}$ such that it preserves classical conditions for truth-value connectives and for each $A\in\F_{\Box\D}$ and $x\in W$ we have:
 {\small \begin{enumerate}[$ - $]\itemsep=0pt
 	\item $ \var(\Box A,x)=\one $ iff $ \forall_{y\in W}\var(A,y)=\one $,
 	\item $ \var(\D A,x)=\one $ iff $ \exists_{y\in W}\var(A,y)=\one $.
 \end{enumerate}}
 \end{definition}
 \begin{definition}
 A formula $ A $ is \textit{true in a world} $w \in W$ iff $ \var(A,w)=T $. A (multi)set of formulas $ \Delta $ \textit{follows} from a (multi)set of formulas $ \Gamma $ in \textbf{S5} ($ \Gamma\models_{\bf S5}\Delta $) iff for every \textbf{S5}-model $ \langle W,\var\rangle $ and every $w \in W$, if every $ B\in\G $ is  true in $w$, then some $ C\in\Delta $ is true in $w$. 
  A formula is \textbf{S5}-\textit{valid} iff it follows from the empty (multi)set of formulas.
 \end{definition}
 \begin{definition} Let $ \langle W,\var\rangle $ be an \textbf{S5}-model. 
 The truth conditions for the two-valued versions of the essence and accident modalities are as follows:
 {\small \begin{enumerate}[$ - $]\itemsep=0pt
 	\item $ \var({\circ} A,x)=\one $ iff $\var(A,x)=\zero$ or $ \forall_{y\in W}\var(A,y)=\one$.
 	\item $ \var({\bullet}A,x)=\one $ iff $\var(A,x)=\one$ and $ \exists_{y\in W}\var(A,y)=\zero$.
 	\item $ \var(\FE A,x)=\one $ iff $\var(A,x)=\one$ or $ \forall_{y\in W}\var(A,y)=\zero$,
 	\item $ \var(\FA A,x)=\one $ iff $\var(A,x)=\zero$ and $ \exists_{y\in W}\var(A,y)=\one$.
 \end{enumerate}}
 \end{definition}
 
 \begin{definition}\label{TruthVal}
 Using the set $ \{T,F\} $ of truth values of classical logic, we can define the set $ \{t,b,n,f\} $ of truth values of \textbf{FDE} and its extensions as follows:  $ \{ \{\one\}, \{\one,\zero\}, \emptyset, \{\zero\} \} $, where $ t $ stands for $ \{\one\} $, $ b $ stands for $ \{\one,\zero\} $, $ n $ for $ \emptyset $, and $ f $ for $ \{\zero\} $. We introduce the following convention: for a formula $ A $ and a valuation $ \var $, we write $ \var(A)\in\{t,b\} $ iff $ T\in \var(A)$ and $ \var(A)\in\{f,b\} $ iff $ F\in \var(A)$.
 \end{definition}

 \begin{definition}\label{KripkeNelsonval}\label{Four-valuedS5}
 A pair $ \langle W,\var\rangle $ is said to be an $ \bf FDE_{\B\D}^{\bf S5}$-model iff $ W\not=\emptyset $  and $ \var $ is a mapping from $ W\times\F_{\Box\D}^\neg $ to $\{t,b,n,f\}$ such that for
 any $A\in\F_{\Box\D}^\neg$ and $x\in W$ we have:
 {\small \begin{enumerate}[$ - $]\itemsep=0pt
 \begin{multicols}{2}
 \item $ \one\in\var(\neg A,x) $ iff $ \zero\in\var(A,x) $,
 \item $ \one\in\var(\non A,x) $ iff $ \one\not\in\var(A,x) $,
 	\item $ \one\in\var(\Box A,x) $ iff $ \forall_{y\in W}\one\in\var(A,y) $,	
 	\item $ \one\in\var(\D A,x) $ iff $ \exists_{y\in W}\one\in\var(A,y) $,
 \item $ \zero\in\var(\neg A,x) $ iff $ \one\in\var(A,x) $,
 \item $ \zero\in\var(\non A,x) $ iff $ \zero\not\in\var(A,x) $,
 \item  $ \zero\in\var(\Box A,x) $ iff $ \exists_{y\in W}\zero\in\var(A,y) $,	
 \item  $ \zero\in\var(\D A,x) $ iff $ \forall_{y\in W}\zero\in\var(A,y) $.
 \end{multicols} \vspace{-.15in}
 \item $ \one\in\var(A\wedge B,x) $ iff $ \one\in\var(A,x) $ and $ \one\in\var(B,x) $,
 \item $ \zero\in\var(A\wedge B,x) $ iff $ \zero\in\var(A,x) $ or $ \zero\in\var(B,x) $,
 \item $ \one\in\var(A\vee B,x) $ iff $ \one\in\var(A,x) $ or $ \one\in\var(B,x) $,
 \item $ \zero\in\var(A\vee B,x) $ iff $ \zero\in\var(A,x) $ and $ \zero\in\var(B,x) $,
 \item $ \one\in\var(A\rightarrow B,x) $ iff $ \one\in\var(A,x) $ implies $ \one\in\var(B,x) $,
 \item $ \zero\in\var(A\rightarrow B,x) $ iff $ \one\in\var(A,x) $ and $ \zero\in\var(B,x) $.
 \end{enumerate}}
 \end{definition}
 
 A similar semantics is considered by Odintsov and Wansing \cite{OdintsovWansing}; there are two main differences: another language (with $ \bot $ and without Boolean negation; the latter has been added later by Sedl\'{a}r \cite{Sedlar}) and another formulation of a valuation (however, `$ x\models^+A $' and `$ x\models^-A $' might be understood as counterparts of `$ T\in \var(A,x) $' and `$ F\in \var(A,x) $' in our semantics). As follows from Definitions \ref{TruthVal} and \ref{Four-valuedS5}, the truth-value conditions for $ \neg $, $ \non $, $ \wedge $, $ \vee $, and $ \rightarrow $ can be presented in the form of the truth tables given below. The truth tables for $ \neg$, $\wedge$, and $\vee$ can be found in \cite{Belnap1,Belnap2}, for $ \rightarrow $ in \cite{ArieliAvron}, and for $ \non $ in \cite{DeOmori}.
 \begin{center}
  \begin{tabular}{|c|c|c|}\hline
 $ A $ & $\neg $ & $ \non $ \\\hline
 \textit{t} & \textit{f} &  \textit{f} \\\hline
 \textit{b} & \textit{b} & \textit{n} \\\hline
 \textit{n} & \textit{n} & \textit{b} \\\hline
 \textit{f} & \textit{t} & \textit{t} \\\hline
 \end{tabular}      
 \begin{tabular}{|c|cccc|}\hline
 $ \vee $ & \textit{t} & \textit{b} & \textit{n} & \textit{f} \\\hline
 \textit{t} & \textit{t} & \textit{t} & \textit{t} & \textit{t} \\\hline
 \textit{b} & \textit{t} & \textit{b} & \textit{t} & \textit{b} \\\hline
 \textit{n} & \textit{t} & \textit{t} & \textit{n} & \textit{n} \\\hline
 \textit{f} & \textit{t} & \textit{b} & \textit{n} & \textit{f} \\\hline
 \end{tabular}     
 \begin{tabular}{|c|cccc|}\hline
 $ \wedge $ & \textit{t} & \textit{b} & \textit{n} & \textit{f} \\\hline
 \textit{t} & \textit{t} & \textit{b} & \textit{n} & \textit{f} \\\hline
 \textit{b} & \textit{b} & \textit{b} & \textit{f} & \textit{f} \\\hline
 \textit{n} & \textit{n} & \textit{f} & \textit{n} & \textit{f} \\\hline
 \textit{f} & \textit{f} & \textit{f} & \textit{f} & \textit{f} \\\hline
 \end{tabular} 
 \begin{tabular}{|c|cccc|}\hline
 $ \rightarrow $ & \textit{t} & \textit{b} & \textit{n} & \textit{f} \\\hline
 \textit{t} & \textit{t} & \textit{b} & \textit{n} & \textit{f} \\\hline
 \textit{b} & \textit{t} & \textit{b} & \textit{n} & \textit{f} \\\hline
 \textit{n} & \textit{t} & \textit{t} & \textit{t} & \textit{t} \\\hline
 \textit{f} & \textit{t} & \textit{t} & \textit{t} & \textit{t} \\\hline
 \end{tabular}           
 \end{center}

 \begin{definition}\label{KripkeBelpanEA}
 A pair $ \langle W,\var\rangle $ is said to be an $ \bf FDE_{\circ\bullet}^{\bf S5}$-model iff $ W\not=\emptyset $  and $ \var $ is a mapping from $ W\times\F_{\circ\bullet}^\neg $ to $\{t,b,n,f\}$ such that for
 any $A\in\F_{\circ\bullet}^\neg$ and $x\in W$ we have 
 \vspace{-.05in}\begin{enumerate}[$ \blacktriangleright $]
 \item the following truth and falsity conditions for the first group of Belnapian essence and accident operators, where $ I\in \{ T,F\} $ and $ \non I $ refers to Boolean negation of $ I $ ($ \non T=F $ and $ \non F=T $):
 \end{enumerate}
 {\small \begin{enumerate}[$ - $]\itemsep=0pt
 \item $ \one\in\var({\circ_I} A,x) $ iff $I\in\var(A,x)$ implies $ \forall_{y\in W}I\in\var(A,y)$,
 \item $ \zero\in\var({\circ_I} A,x) $ iff $I\in\var(A,x)$ and $ \exists_{y\in W}\non I\in\var(A,y)$,

 \item $ \one\in\var(\FE_I A,x) $ iff $I\not\in\var(A,x)$ implies $ \forall_{y\in W}I\not\in\var(A,y)$,
 \item $ \zero\in\var(\FE_I A,x) $ iff $I\not\in\var(A,x)$ and $ \exists_{y\in W}\non I\not\in\var(A,y)$,

 \item $ \one\in\var({\bullet_I}A,x) $ iff $I\in\var(A,x)$ and $ \exists_{y\in W}I\not\in\var(A,y)$,
 \item  $ \zero\in\var({\bullet_I} A,x) $ iff $\non I\not\in\var(A,x)$ implies $ \forall_{y\in W}\non I\not\in\var(A,y)$,
 \item $\one\in \var(\FA_I A,x) $ iff $I\not\in\var(A,x)$ and $ \exists_{y\in W}I\in\var(A,y)$,
 \item $ \zero\in\var(\FA_I A,x) $ iff $\non I\in\var(A,x)$ implies $ \forall_{y\in W}\non I\in\var(A,y)$.
 \end{enumerate}}
 
 \vspace{-.05in}\begin{enumerate}[$ \blacktriangleright $]
 \item the following truth conditions for the second group of Belnapian essence and accident operators, where $ i\in\{t,b,n,f\} $ and $ \ast i $ refers to the conflation of these values ($ \ast t=t,\ast b=n,\ast n=b,\ast f=f $):
 \end{enumerate}
 \begin{enumerate}[$ - $] \itemsep=0.01pt
 {\small \item $ \one\in\var({\circ_i} A,x) $ iff  $\var(A,x)=i$ implies $ \forall_{y\in W}\var(A,y)=i $,
 
 \item $ \zero\in\var({\circ_i} A,x) $ iff $\var(A,x)=i$ and $ \exists_{y\in W}\var(A,y)\not=\ast i$,}
 
 {\small \item $ \one\in\var({\FE_i} A,x) $ iff  $\var(A,x)\not=i$ implies $ \forall_{y\in W}\var(A,y)\not=i $,
 
 \item  $ \zero\in\var({\FE_i} A,x) $ iff $\var(A,x)\not=i$ and $ \exists_{y\in W}\var(A,y)=\ast i$.}
 {\small \item $ \one\in\var({\bullet_i}A,x) $ iff $\var(A,x)=i$ and $ \exists_{y\in W}\var(A,y)\not= i$,
 \item $ \zero\in\var({\bullet_i} A,x) $ iff $\var(A,x)=\ast i$ implies $ \forall_{y\in W}\var(A,y)=\ast i$,}
 {\small \item $ \one\in\var({\FA_i}A,x) $ iff $\var(A,x)\not=i$ and $ \exists_{y\in W}\var(A,y)= i$,
 	\item  $ \zero\in\var({\FA_i} A,x) $ iff $\var(A,x)\not=\ast i$ implies $ \forall_{y\in W}\var(A,y)\not=\ast i$.}		
 \end{enumerate}
 \end{definition}

 \begin{definition}
 In $ \bf FDE_{\B\D}^{\bf S5}$ and $ \bf FDE_{\circ\bullet}^{\bf S5}$, $ A $ is \textit{at least true in a world} $w \in W$ iff $ T\in\var(A,w) $.  In $ \bf FDE_{\B\D}^{\bf S5}$ and $ \bf FDE_{\circ\bullet}^{\bf S5}$, a (multi)set of formulas $ \Delta $ \textit{follows} from a (multi)set of formulas $ \Gamma $ iff for every $ \bf FDE_{\B\D}^{\bf S5}$-model (resp. $ \bf FDE_{\circ\bullet}^{\bf S5}$-model) $ \langle W,\var\rangle $ and every $w \in W$, if every $ B\in\G $ is at least true in $w$, then some $ C\in\Delta $ is at least true in $w$. Let $ {\bf L}\in \{\bf FDE_{\B\D}^{\bf S5},FDE_{\circ\bullet}^{\bf S5} \} $. A formula is \textbf{L}-\textit{valid} iff it follows from the empty (multi)set of formulas.
 \end{definition}
 
 \begin{prop}
 Let $ \langle W,\var\rangle $ be an $ \bf FDE_{\circ\bullet}^{\bf S5}$-model. 
 Let $ I\in\{T,F\} $, $ i\in\{t,b,n,f\} $, and $ j\in\{t,f\} $.
 For each formula $ A $ and each $ x\in W $, it holds that:
 {\small\begin{enumerate}[$ \blacktriangleright $]\itemsep=0pt
 \item $ \one\in\var({\circ_I}A,x)$ iff $\one\in\var(\non{\bullet_I}A,x) $ iff $\one\in\var({\FE_I}\non A,x) $ iff $\one\in\var(\non{\FA_I}\non A,x) $,
 \item $ \one\in\var({\bullet_I}A,x)$ iff $\one\in\var(\non{\circ_I}A,x) $ iff $\one\in\var({\FA_I}\non A,x) $ iff $\one\in\var(\non{\FE_I}\non A,x) $,
 \item $ \one\in\var({\FE_I}A,x)$ iff $\one\in\var(\non{\FA_I}A,x) $ iff $\one\in\var({\circ_I}\non A,x) $ iff $\one\in\var(\non{\bullet_I}\non A,x) $,
 \item $ \one\in\var({\FA_I}A,x)$ iff $\one\in\var(\non{\FE_I}A,x) $ iff $\one\in\var({\bullet_I}\non A,x) $ iff $\one\in\var(\non{\circ_I}\non A,x) $,
 
 \item $ \one\in\var({\circ_i}A,x)$ iff $\one\in\var(\non{\bullet_i}A,x) $ iff $\one\in\var(\neg{\bullet_{*i}}A,x) $, 
 \item $ \one\in\var({\FE_i}A,x)$ iff $\one\in\var(\non{\FA_i}A,x) $ iff $\one\in\var(\neg{\FA_{*i}}A,x) $, 
 
 \vspace{-.1in}\begin{multicols}{2}
 \item $ \one\in\var({\bullet_i}A,x)$ iff $\one\in\var(\non{\circ_i}A,x) $,
 \item $ \one\in\var({\FA_i}A,x)$ iff $\one\in\var(\non{\FE_i}A,x) $, 
 \item $ \one\in\var({\bullet_j}A,x)$ iff $\one\in\var(\neg{\circ_{j}}A,x) $, 
 \item $ \one\in\var({\FA_j}A,x)$ iff $\one\in\var(\neg{\FE_{j}}A,x) $.
 \end{multicols}
 \vspace{-.15in}
 \end{enumerate}}
 \end{prop}
 \begin{proof}
 Follows from Definitions \ref{KripkeNelsonval} and \ref{KripkeBelpanEA}.
 \end{proof}
 
 \begin{prop}
 Let $ \langle W,\var\rangle $ be an $ \bf FDE_{\circ\bullet}^{\bf S5}$-model. 
 Let $ \star\in\{\circ,\bullet,\FE,\FA\} $. 
 For each formula $ A $ and each $ x\in W $, it holds that:
 {\small \begin{enumerate}[$ \blacktriangleright $]\itemsep=0pt
 \begin{multicols}{2}
 \item $\one\in\var(\star_T A,x)$ iff $\one\in\var(\star_F\neg A,x)$,
 \item $\one\in\var(\star_F A,x)$ iff $\one\in\var(\star_T\neg A,x)$,
 \end{multicols}\vspace{-.15in}
 \item $\one\in\var(\star_t A,x)$ iff $\one\in\var(\star_f\neg A,x)$ iff $\one\in\var(\star_f\non A,x)$,
 \item $\one\in\var(\star_f A,x)$ iff $\one\in\var(\star_t\neg A,x)$ iff $\one\in\var(\star_t\non A,x)$,
 \item $\one\in\var(\star_b A,x)$ iff $\one\in\var(\star_b\neg A,x)$ iff $\one\in\var(\star_n\non A,x)$,
 \item $\one\in\var(\star_n A,x)$ iff $\one\in\var(\star_n\neg A,x)$ iff $\one\in\var(\star_b\non A,x)$.
 \end{enumerate}}
 \end{prop}
 \begin{proof}
 Follows from Definitions \ref{KripkeNelsonval} and \ref{KripkeBelpanEA}.
 \end{proof}
 
 \begin{definition}\label{i} For any $\Lang_{\circ\bullet}^\neg$-formula $ A $, we define the following abbreviations:
 \begin{center}
 $A^T=A $,\qquad\qquad  $A^F=\neg A$,\\
 $ A^t= A\wedge\non \neg A$,\qquad $ A^b=A\wedge \neg A $,\qquad $ A^n =\non A\wedge\non \neg A$,\qquad $ A^f =\non A\wedge \neg A$.
 \end{center}
 \end{definition}
 
 A similar definition was employed in \cite{DeOmori,PetrukhinFDE,Sedlar}. Informally, $ A $ means that it has value $t$ or $b$, $ \non A$ means that $ A $ has value $f$ or $n$, $ \neg A$ means that $ A $ has value $f$ or $b$, and $ \non\neg A$ signifies that $ A $ has value $t$ or $n$. Hence, for example, $ A\wedge\non \neg A $ means that $ A $ takes the value $ t $ or $ b $ and $ A $ does not take values $ f $ and $ b $: thus, $ A $ takes the value $ t $. 
 \begin{lemma}\label{X}
 Let $ I\in\{T,F\} $ and $ i\in\{t,b,n,f \} $. For every formula $ A $, every valuation $\var$, and every world $ x $, it holds that
 \begin{center}
 {\small\begin{enumerate}[$ \blacktriangleright $]\itemsep=0pt
 \begin{multicols}{2}
 \item $ T\in\var(A^I,x) $ iff $ I\in\var(A,x) $,
 \item $ T\not\in\var(A^I,x) $ iff $ I\not\in\var(A,x) $,
 \item $ F\in\var(A^I,x) $ iff $  I\in\var(\neg A,x) $,
 \item $ F\not\in\var(A^I,x) $ iff $  I\not\in\var(\neg A,x) $,
 \item $ T\in\var(A^i,x) $ iff $ \var(A,x)=i $,
 \item $ T\not\in\var(A^i,x) $ iff $ \var(A,x)\not=i $,
 \item $ F\in\var(A^i,x) $ iff $ \var(A,x)\not=\ast i $,
 \item $ F\not\in\var(A^i,x) $ iff $ \var(A,x)=\ast i $.
 \end{multicols}
 \end{enumerate}}
 \end{center}
 \end{lemma}
 \begin{proof}
 Follows from Definition \ref{i}.
 \end{proof}
 
 \begin{prop}
 Let $ \langle W,\var\rangle $ be an $ \bf FDE_{\circ\bullet}^{\bf S5}$-model. 
 Let $ \star\in\{\circ,\bullet,\FE,\FA\} $, $i\in\{t,b,n,f\}$, and $I\in\{T,F\}$. 
 For each formula $ A $ and each $ x\in W $, it holds that:
 {\small \begin{enumerate}[$ \blacktriangleright $]\itemsep=0pt
 \begin{multicols}{2}
 \item $\one\in\var(\star_i A,x)$ iff $\one\in\var(\star_T A^i,x)$,
 \item $\one\in\var(\star_I A,x)$ iff $\one\in\var(\star_T A^I,x)$,
 \end{multicols}\vspace{-.17in}
 \end{enumerate}}
 \end{prop}
 \begin{proof}
 Follows from Definitions \ref{KripkeNelsonval}, \ref{KripkeBelpanEA}, and \ref{i}. Use Lemma \ref{X}.
 \end{proof}
 
 \begin{hypo}
 For every $\star\in\{\circ,\bullet,\FE,\FA\}$ and every $J\in\{t,f,T,F\}$,
 the operator $\star_J$ is not definable in terms of $\star_b$ and $\star_n$ (and propositional connectives).
 \end{hypo}
 
 In what follows, we work with the full family of 24 essence and accident operators (for values $ i\in\{t,b,n,f\} $ and for the ``at least'' components $ I\in\{T,F\} $). Although many of these operators are definable from others and from $\Box$ and $\Diamond$, we deliberately keep them as primitive in order to obtain a uniform semantic and proof-theoretic treatment of the different truth-statuses.

\section{Hypersequent calculi}\label{Hypersequent}
\begin{definition}
A sequent is an ordered pair written as $\Gamma \Rightarrow \Delta$, where $\Gamma$ and $\Delta$ are finite multisets of formulas. A hypersequent is a finite multiset of sequents, written as $\Gamma_1 \Rightarrow \Delta_1 \mid \ldots \mid \Gamma_n \Rightarrow \Delta_n$. A sequent $\Gamma \Rightarrow \Delta$ is valid in a logic $\mathbf{L}$ iff $\Gamma \models_{\mathbf{L}} \Delta$. A hypersequent is valid iff at least one of its component sequents is valid. For a more precise definition via a translation of hypersequents into formulas, see \cite[p.~211]{IndrzejczakSC}.
\end{definition}

Let us describe a slightly modified version of Restall's hypersequent calculus for \textbf{S5} \cite[pp. 163, 164]{Restall}. The differences are discussed in Remark~\ref{RestallRemark} below.

\begin{definition}[Restall's hypersequent calculus \cite{Restall} for \textbf{S5}]
    The calculus has the  axiom  (Ax) $ p\Rightarrow p$, for each $p\in P$, and the following structural rules, where $i\in \{1,2\}$: 
    \begin{center}
{\small 	{\small (EW)} $ \dfrac{H_i}{H_1\mid H_2} $\qquad}	
{\small 	{\small (IC$\Rightarrow$)} $ \dfrac{A,A,\Gamma\Rightarrow\Delta\mid H}{A,\Gamma\Rightarrow\Delta\mid H} $\qquad				
	{\small ($\Rightarrow$IC)} $ \dfrac{\Gamma\Rightarrow\Delta,A,A\mid H}{\Gamma\Rightarrow\Delta,A\mid H} $\quad}									
\end{center}
\begin{center}				
{\small 	{\small (Cut)} $ \dfrac{\Gamma\Rightarrow\Delta,A\mid H\qquad A,\Theta\Rightarrow\Lambda\mid G}{\Gamma,\Theta\Rightarrow\Delta,\Lambda\mid H\mid G} $\qquad									
		{\small (Merge)} $ \dfrac{\Gamma\Rightarrow\Delta\mid\Theta\Rightarrow\Lambda\mid H}{\Gamma,\Theta\Rightarrow\Delta,\Lambda\mid H} $	}
\end{center}

The rules for truth-value connectives are as follows:

\begin{center}				
{\small 	{\small $ (\non\Rightarrow) $} $ \dfrac{\Gamma\Rightarrow\Delta,A\mid H}{\non A,\Gamma\Rightarrow\Delta\mid H} $\qquad				
	{\small $ (\Rightarrow\non) $} $ \dfrac{A,\Gamma\Rightarrow\Delta\mid H}{\Gamma\Rightarrow\Delta,\non A\mid H} $	}			
\end{center}
\begin{center}				
{\small 	{\small $ (\wedge\Rightarrow) $} $ \dfrac{A,B,\Gamma\Rightarrow\Delta\mid H}{A\wedge B,\Gamma\Rightarrow\Delta\mid H} $\qquad				
	{\small $ (\Rightarrow\wedge) $} $ \dfrac{\Gamma\Rightarrow\Delta,A\mid H\qquad\Theta\Rightarrow\Lambda,B\mid G}{\Gamma,\Theta\Rightarrow\Delta,\Lambda,A\wedge B\mid H\mid G} $	}			
\end{center}				
\begin{center}				
{\small 	{\small $ (\vee\Rightarrow) $} $ \dfrac{A,\Gamma\Rightarrow\Delta\mid H\qquad B,\Theta\Rightarrow\Lambda\mid G}{A\vee B,\Gamma,\Theta\Rightarrow\Delta,\Lambda\mid H\mid G} $\qquad				
	{\small $ (\Rightarrow\vee) $} $ \dfrac{\Gamma\Rightarrow\Delta,A,B\mid H}{\Gamma\Rightarrow\Delta,A\vee B\mid H} $\qquad	}			
\end{center}				
\begin{center}				
{\small 	{\small $ (\rightarrow\,\Rightarrow) $} $ \dfrac{\Gamma\Rightarrow\Delta,A\mid H\qquad B,\Theta\Rightarrow\Lambda\mid G}{A\rightarrow B,\Gamma,\Theta\Rightarrow\Delta,\Lambda\mid H\mid G} $\qquad				
	{\small $ (\Rightarrow\,\rightarrow) $} $ \dfrac{A,\Gamma\Rightarrow\Delta,B\mid H}{\Gamma\Rightarrow\Delta,A\rightarrow B\mid H} $	}			
\end{center}

The rules for necessity and possibility  are as follows:
\begin{center}				
{\small 	{\small $ (\Box\Rightarrow) $} $ \dfrac{A,\Gamma\Rightarrow\Delta\mid H}{\Box A\Rightarrow\;\mid \Gamma\Rightarrow\Delta\mid H} $\qquad				
	{\small $ (\Rightarrow\Box) $} $ \dfrac{\Rightarrow A\mid H}{\Rightarrow\Box A\mid H} $}				
\end{center}
\begin{center}				
{\small 	{\small $ (\Diamond\Rightarrow) $} $ \dfrac{A\Rightarrow\;\mid H}{\Diamond A\Rightarrow\;\mid H} $\qquad				
	{\small $ (\Rightarrow\Diamond) $} $ \dfrac{\Gamma\Rightarrow\Delta,A\mid H}{\Gamma\Rightarrow\Delta\mid\; \Rightarrow\Diamond A\mid H} $	}			
\end{center} 
\end{definition}
\begin{remark}\label{RestallRemark}
A few comments are in order:   
\begin{enumerate}[$(1)$]\itemsep=0pt
\item Following \cite{MMLS5}, we extend Restall's original formulation \cite{Restall} by adding rules for  $\vee$,  $\rightarrow$, and $\D $. The rules for $ \wedge $ and $ \vee $ are given in multiplicative form in order to shorten the proofs of the embedding theorems. 
\item The axiom (Ax) $ p\Rightarrow p$, for each $p\in P$, can be generalised for the case of arbitrary formulas by induction:  $ A\Rightarrow A$. Such a generalised version is used in \cite{Restall}.
\item The original formulation in \cite{Restall} uses the following rules of external weakening:
\begin{center}
{\small (EW$\Rightarrow$) $ \dfrac{H}{A\Rightarrow\;\mid H} $\qquad				
($\Rightarrow$EW) $ \dfrac{H}{\Rightarrow A\mid H} $\qquad}
\end{center}
However, as shown in \cite{PetrukhinS5}, in order to obtain a constructive cut-elimination proof for Restall's calculus one has to use a more general version of external weakening, which allows us to add not only sequents of the form $ A\Rightarrow $ or $ \Rightarrow A$, but arbitrary hypersequents, including the empty one. Although we do not deal with constructive cut elimination in this paper, we adopt this more general formulation of external weakening.
    \item The internal weakening and external contraction rules are admissible:  
\begin{center}
{\small 	{\small (IW$\Rightarrow$)} $ \dfrac{\Gamma\Rightarrow\Delta\mid H}{A,\Gamma\Rightarrow\Delta\mid H} $\qquad				
	{\small ($\Rightarrow$IW)} $ \dfrac{\Gamma\Rightarrow\Delta\mid H}{\Gamma\Rightarrow\Delta, A\mid H} $\qquad
	{\small (EC)} $ \dfrac{\Gamma\Rightarrow\Delta\mid\Gamma\Rightarrow\Delta\mid H}{\Gamma\Rightarrow\Delta\mid H} $\qquad	}									\end{center}

\item As pointed out by a reviewer, the rule (Merge) is not invertible, since the hypersequent $p \Rightarrow \mid \Rightarrow p$ is not valid, whereas $p \Rightarrow p$ is valid. However, this does not affect the cut-elimination proof for this calculus given in \cite{Restall} (see also \cite{IndrzejczakSC}). In what follows, we prove cut elimination for $\mathbf{FDE}^{\bf S5}_{\Box\D}$ and $\mathbf{FDE}^{\bf S5}_{\circ\bullet}$ using embedding procedures that do not require the rules to be invertible.
\item The notion of a proof in this hypersequent calculus 
is defined in the standard way. 
\item A modification of Restall's hypersequent calculus for \textbf{S5} that formalises the two-valued essence and accident modalities as primitive connectives is presented in \cite{PetrukhinS5}.
\end{enumerate}
\end{remark}

\begin{definition}
A hypersequent calculus for $\mathbf{FDE}_{\B\D}^{\bf S5}$ is obtained from Restall's hypersequent calculus for \textbf{S5} by adding the axiom (Ax$_\neg$) $\neg p \Rightarrow \neg p$, for each $p \in \Prop$, and the rules given below. The rules for negated propositional connectives are as follows:

\begin{center}				
{\small $ (\neg\wedge\Rightarrow) $ $ \dfrac{\neg\A,\Gamma\Rightarrow\Delta\mid H\qquad\neg\Be,\Theta\Rightarrow\Lambda\mid G}{\neg(\A\wedge\Be),\Gamma,\Theta\Rightarrow\Delta,\Lambda\mid H\mid G} $\qquad				
$ (\Rightarrow\neg\wedge) $ $ \dfrac{\Gamma\Rightarrow\Delta,\neg\A,\neg\Be\mid H}{\Gamma\Rightarrow\Delta,\neg(\A\wedge\Be)\mid H} $	}			
\end{center}				
\begin{center}				
{\small $ (\neg\vee\Rightarrow) $ $ \dfrac{\neg\A,\neg\Be,\Gamma\Rightarrow\Delta\mid H}{\neg(\A\vee\Be),\Gamma\Rightarrow\Delta\mid H} $\qquad				
$ (\Rightarrow\neg\vee) $ $ \dfrac{\Gamma\Rightarrow\Delta,\neg\A\mid H\qquad\Theta\Rightarrow\Lambda,\neg\Be\mid G}{\Gamma,\Theta\Rightarrow\Delta,\Lambda,\neg(\A\vee\Be)\mid H\mid G} $	}			
\end{center}				
\begin{center}				
{\small $ (\neg{\rightarrow}\Rightarrow) $ $ \dfrac{\A,\neg\Be,\Gamma\Rightarrow\Delta\mid H}{\neg(\A\rightarrow\Be),\Gamma\Rightarrow\Delta\mid H} $\qquad				
$ (\Rightarrow\neg{\rightarrow}) $ $ \dfrac{\Gamma\Rightarrow\Delta,\A\mid H\qquad\Theta\Rightarrow\Lambda,\neg\Be\mid G}{\Gamma,\Theta\Rightarrow\Delta,\Lambda,\neg(\A\rightarrow\Be)\mid H\mid G} $}				
\end{center}				

\begin{center}				
{\small $ (\neg\neg\Rightarrow) $ $ \dfrac{\A,\Gamma\Rightarrow\Delta\mid H}{\neg\neg\A,\Gamma\Rightarrow\Delta\mid H} $\qquad				
$ (\Rightarrow\neg\neg) $ $ \dfrac{\Gamma\Rightarrow\Delta,\A\mid H}{\Gamma\Rightarrow\Delta,\neg\neg\A\mid H} $	}			
\end{center}
\begin{center}				
{\small $ (\neg\non\Rightarrow) $ $ \dfrac{\Gamma\Rightarrow\Delta,\neg A\mid H}{\neg\non\A,\Gamma\Rightarrow\Delta\mid H} $\qquad				
$ (\Rightarrow\neg\non) $ $ \dfrac{\neg A,\Gamma\Rightarrow\Delta\mid H}{\Gamma\Rightarrow\Delta,\neg\non\A\mid H} $}				
\end{center}

Rules for negated modalities are as follows:

\begin{center}				
{\small $ (\neg\Box\Rightarrow) $ $ \dfrac{\neg A\Rightarrow\;\mid H}{\neg\Box A\Rightarrow\;\mid H} $\qquad				
$ (\Rightarrow\neg\Box) $ $ \dfrac{\Gamma\Rightarrow\Delta,\neg A\mid H}{\Gamma\Rightarrow\Delta\mid\;\Rightarrow\neg\Box A\mid H} $}				
\end{center}				
\begin{center}				
{\small $ (\neg\Diamond\Rightarrow) $ $ \dfrac{\neg A,\Gamma\Rightarrow\Delta\mid H}{\neg\Diamond A\Rightarrow\;\mid\Gamma\Rightarrow\Delta\mid H} $\qquad				
$ (\Rightarrow\neg\Diamond) $ $ \dfrac{\Rightarrow\neg A\mid H}{\Rightarrow\neg\Diamond A\mid H} $}				
\end{center}
\end{definition}
\begin{remark}
  The rules for negated conjunction, disjunction, and implication as well as for double negation were introduced in \cite{AvronClassical,Popov} and were used also in \cite{Kamide19}, where the rules for $ \neg\non $ were presented. The rules for conjunction, disjunction, and implication are given in multiplicative form in order to shorten the proofs of the embedding theorems. The negated modal rules are due to \cite{Kamide23}, and, as noted there, are an adaptation of the multilattice rules from \cite{MMLS5}.  The notion of a proof in this hypersequent calculus is defined in the standard way. 
\end{remark}
\begin{definition}
A hypersequent calculus for $\mathbf{FDE}_{\circ\bullet}^{\bf S5}$ is obtained from the calculus for $\mathbf{FDE}_{\B\D}^{\bf S5}$ by replacing the rules for $\B$ and $\D$ with the rules for the essence and accident modalities. The hypersequent rules for the first group of Belnapian essence operators are as follows:
\begin{center}
{\small 	{\small $ (\circ_T\Rightarrow) $} $ \dfrac{A,\G\Rightarrow \Delta\mid H\qquad\Theta\Rightarrow \Lambda,A\mid G}{{\circ}_TA,\Theta\Rightarrow\Lambda\mid \Gamma\Rightarrow\Delta\mid H\mid G} $\qquad
	{\small $ (\Rightarrow\circ_T) $} $ \dfrac{\Rightarrow A\mid A,\G\Rightarrow \Delta\mid H}{\G\Rightarrow \Delta,{\circ}_TA\mid H} $}
\end{center}
\begin{center}
{\small 	{\small $ (\neg\circ_T\Rightarrow) $}  $ \dfrac{\neg A\Rightarrow\;\mid A,\G\Rightarrow\Delta\mid H}{\neg{\circ_T}A,\G\Rightarrow \Delta\mid H} $ \qquad
	{\small $ (\Rightarrow\neg\circ_T) $} 
$ \dfrac{\Gamma\Rightarrow \Delta,\neg A\mid H\qquad\Theta\Rightarrow \Lambda,A\mid G}{\Theta\Rightarrow \Lambda,\neg{\circ_T}A\mid\Gamma\Rightarrow \Delta\mid H\mid G } $	}
\end{center}
\begin{center}
{\small {\small $ (\FE_T \Rightarrow) $} $ \dfrac{A,\Gamma\Rightarrow \Delta\mid H \qquad \Theta\Rightarrow \Lambda,A\mid G}{\FE_TA,\Gamma\Rightarrow \Delta\mid \Theta\Rightarrow \Lambda\mid H\mid G} $\qquad
{\small $ (\Rightarrow\FE_T) $} $ \dfrac{\G\Rightarrow\Delta, A\mid A\Rightarrow\; \mid H}{\G\Rightarrow \Delta,\FE_TA\mid H} $}
\end{center}
\begin{center}
{\small 	{\small $ (\neg\FE_T\Rightarrow) $}  $ \dfrac{\Rightarrow\neg A\mid \G\Rightarrow\Delta,A\mid H}{\neg\FE_TA,\G\Rightarrow \Delta\mid H} $ \qquad
	{\small $ (\Rightarrow\neg\FE_T) $} $ \dfrac{A,\Gamma\Rightarrow \Delta\mid H\qquad\neg A,\Theta\Rightarrow \Lambda\mid G}{\Gamma\Rightarrow \Delta,\neg{\FE_T}A\mid\Theta\Rightarrow \Lambda\mid H\mid G } $}
\end{center}
\begin{center}
{\small {\small $ (\circ_F \Rightarrow) $} $ \dfrac{\neg A,\Gamma\Rightarrow \Delta\mid H \qquad \Theta\Rightarrow \Lambda,\neg A\mid G}{\circ_FA,\Theta\Rightarrow \Lambda\mid \Gamma\Rightarrow \Delta\mid H\mid G} $\qquad
{\small $ (\Rightarrow\circ_F) $ } $ \dfrac{\Rightarrow\neg A\mid \neg A,\G\Rightarrow \Delta\mid H}{\G\Rightarrow \Delta,\circ_FA\mid H} $}
\end{center}
 \begin{center}
{\small 	{\small $ (\neg\circ_F\Rightarrow) $} $ \dfrac{A\Rightarrow\;\mid \neg A,\G\Rightarrow\Delta\mid H}{\neg{\circ_F}A,\G\Rightarrow \Delta\mid H} $ \qquad
	{\small $ (\Rightarrow\neg\circ_F) $} $ \dfrac{\Gamma\Rightarrow \Delta, A\mid H\qquad\Theta\Rightarrow \Lambda,\neg A\mid G}{\Theta\Rightarrow \Lambda,\neg{\circ_F}A\mid\Gamma\Rightarrow \Delta\mid H\mid G } $}
\end{center}
\begin{center}
{\small {\small $ (\FE_F \Rightarrow) $} $ \dfrac{\neg A,\Gamma\Rightarrow \Delta\mid H \qquad \Theta\Rightarrow \Lambda,\neg A\mid G}{\FE_FA,\Gamma\Rightarrow \Delta\mid \Theta\Rightarrow \Lambda\mid H\mid G} $\qquad
{\small $ (\Rightarrow\FE_F) $} $ \dfrac{\G\Rightarrow\Delta,\neg A\mid \neg A\Rightarrow \mid H}{\G\Rightarrow \Delta,{\FE}_FA\mid H} $}
\end{center}
\begin{center}
{\small 	{\small $ (\neg{\FE}_F\Rightarrow) $} $ \dfrac{\Rightarrow A\mid \G\Rightarrow\Delta,\neg A\mid H}{\neg\FE_FA,\G\Rightarrow \Delta\mid H} $ \qquad
	{\small $ (\Rightarrow\neg{\FE}_F) $}  $ \dfrac{\neg A,\Gamma\Rightarrow \Delta\mid H\qquad A,\Theta\Rightarrow \Lambda\mid G}{\Gamma\Rightarrow \Delta,\neg{\FE_F}A\mid\Theta\Rightarrow \Lambda\mid H\mid G } $}
\end{center}

Hypersequent rules for the first group of Belnapian accident operators:
\begin{center}
{\small 	{\small $ (\bullet_T\Rightarrow) $} $ \dfrac{\Rightarrow A\mid A,\G\Rightarrow \Delta\mid H}{{\bullet}_TA,\G\Rightarrow \Delta\mid H} $ \qquad
	{\small $ (\Rightarrow\bullet_T) $} $ \dfrac{A,\G\Rightarrow \Delta\mid H\qquad\Theta\Rightarrow \Lambda,A\mid G}{\Theta\Rightarrow\Lambda,{\bullet}_TA\mid \Gamma\Rightarrow\Delta\mid H\mid G} $}
\end{center}
\begin{center}
{\small 	{\small $ (\neg{\bullet_T}\Rightarrow) $} 
$ \dfrac{\neg A,\Gamma\Rightarrow \Delta\mid H\qquad\Theta\Rightarrow \Lambda,\neg A\mid G}{\neg{\bullet_T}A,\Gamma\Rightarrow \Delta\mid\Theta\Rightarrow \Lambda\mid H\mid G } $\qquad
	{\small $ (\Rightarrow\neg{\bullet_T}) $}  $ \dfrac{\neg A\Rightarrow\;\mid \G\Rightarrow\Delta,\neg A\mid H}{\G\Rightarrow \Delta,\neg{\bullet_T}A\mid H} $ \qquad}
\end{center}
\begin{center}
{\small {\small $ (\FA_T\Rightarrow) $} $ \dfrac{\G\Rightarrow\Delta, A\mid A\Rightarrow\; \mid H}{\FA_TA,\G\Rightarrow \Delta\mid H} $\qquad
{\small $ (\Rightarrow\FA_T) $} $ \dfrac{A,\Gamma\Rightarrow \Delta\mid H \qquad \Theta\Rightarrow \Lambda,A\mid G}{\Gamma\Rightarrow \Delta,\FA_TA\mid \Theta\Rightarrow \Lambda\mid H\mid G} $}
\end{center}
\begin{center}
{\small 	{\small $ (\neg{\FA_T}\Rightarrow) $} $ \dfrac{\neg A,\Gamma\Rightarrow \Delta\mid H\qquad\Theta\Rightarrow \Lambda,\neg A\mid G}{\neg{\FA_T}A,\Theta\Rightarrow \Lambda\mid\Gamma\Rightarrow \Delta\mid H\mid G } $\qquad
	{\small $ (\Rightarrow\neg{\FA_T}) $}  $ \dfrac{\Rightarrow\neg A\mid \neg A,\G\Rightarrow\Delta\mid H}{\G\Rightarrow \Delta,\neg{\FA_T}A\mid H} $ \qquad	}
\end{center}
\begin{center}
{\small {\small $ (\bullet_F\Rightarrow) $ } $ \dfrac{\Rightarrow\neg A\mid \neg A,\G\Rightarrow \Delta\mid H}{\bullet_FA,\G\Rightarrow \Delta\mid H} $\qquad
{\small $ (\Rightarrow\bullet_F) $} $ \dfrac{\neg A,\Gamma\Rightarrow \Delta\mid H \qquad \Theta\Rightarrow \Lambda,\neg A\mid G}{\Theta\Rightarrow \Lambda,\bullet_FA\mid \Gamma\Rightarrow \Delta\mid H\mid G} $\qquad}
\end{center}
\begin{center}
{\small 	{\small $ (\neg\bullet_F\Rightarrow) $} $ \dfrac{A,\Gamma\Rightarrow \Delta\mid H\qquad\Theta\Rightarrow \Lambda, A\mid G}{\neg{\bullet_F}A,\Gamma\Rightarrow \Delta\mid\Theta\Rightarrow \Lambda\mid H\mid G } $\qquad
	{\small $ (\Rightarrow\neg{\bullet_F}) $} $ \dfrac{A\Rightarrow\;\mid \G\Rightarrow\Delta,A\mid H}{\G\Rightarrow \Delta,\neg{\bullet_F}A\mid H} $ \qquad	}
\end{center}
\begin{center}
{\small {\small $ (\FA_F\Rightarrow) $} $ \dfrac{\G\Rightarrow\Delta,\neg A\mid \neg A\Rightarrow\;\mid H}{{\FA}_FA,\G\Rightarrow \Delta\mid H} $\qquad
{\small $ (\Rightarrow\FA_F) $} $ \dfrac{\neg A,\Gamma\Rightarrow \Delta\mid H \qquad \Theta\Rightarrow \Lambda,\neg A\mid G}{\Gamma\Rightarrow \Delta,\FA_F A\mid \Theta\Rightarrow \Lambda\mid H\mid G} $\qquad}
\end{center}
\begin{center}
{\small {\small $ (\neg\FA_F\Rightarrow) $} $ \dfrac{A,\Gamma\Rightarrow \Delta\mid H \qquad \Theta\Rightarrow \Lambda,A\mid G}{\neg \FA_F A,\Theta\Rightarrow \Lambda\mid\Gamma\Rightarrow \Delta\mid  H\mid G} $\qquad
{\small $ (\Rightarrow\neg \FA_F) $} $ \dfrac{A,\G\Rightarrow\Delta\mid \;\Rightarrow A\mid H}{\G\Rightarrow \Delta,\neg {\FA}_FA\mid H} $\qquad}
\end{center}

Hypersequent rules for the second group of Belnapian essence operators are as follows (where $ S_1 $ stands for $ \Gamma\Rightarrow\Delta $ and $ S_2 $ stands for $ \Theta\Rightarrow\Lambda $):
{\footnotesize \begin{center}
$ (\circ_t\Rightarrow) $ $ \dfrac{A,S_1,\neg A\mid H\quad S_2,A\mid G\quad\neg A,S_2\mid G}{{\circ_t}A,S_2\mid S_1\mid H\mid G} $\quad
$ (\Rightarrow\circ_t) $ $ \dfrac{\Rightarrow A\mid A,S_1,\neg A\mid H\quad\neg A\Rightarrow\;\mid A,S_1,\neg A\mid G}{S_1,{\circ_t}A\mid H\mid G} $
\end{center}
\begin{center}
$ (\neg\circ_t\Rightarrow) $ $ \dfrac{\Rightarrow A\mid A,S_1,\neg A\mid H\quad\neg A\Rightarrow\;\mid A,S_1,\neg A\mid G}{\neg{\circ_t}A,S_1\mid H\mid G} $\quad
$ (\Rightarrow\neg\circ_t) $ $ \dfrac{A,S_1,\neg A\mid H\quad S_2,A\mid G\quad\neg A,S_2\mid G}{S_2,\neg{\circ_t}A\mid S_1\mid H\mid G} $
\end{center}

\begin{center}
$ (\circ_b\Rightarrow) $ $ \dfrac{A,\neg A,S_1\mid H\quad S_2,A\mid G\quad S_2,\neg A\mid G}{{\circ_b}A,S_2\mid S_1\mid H\mid G} $\quad
$ (\Rightarrow\circ_b) $ $ \dfrac{\Rightarrow A\mid A,\neg A,S_1\mid H\quad\Rightarrow\neg A\mid A,\neg A,S_1\mid G}{S_1,{\circ_b}A\mid H\mid G} $
\end{center}
\begin{center}
$ (\neg\circ_b\Rightarrow) $ $ \dfrac{A\Rightarrow \mid A,\neg A,S_1\mid H\quad\neg A\Rightarrow\;\mid A,\neg A,S_1\mid G}{\neg{\circ_b}A,S_1\mid H\mid G} $\quad
$ (\Rightarrow\neg\circ_b) $ $ \dfrac{S_1,A,\neg A\mid H\quad S_2,A\mid G\quad S_2,\neg A\mid G}{S_2,\neg{\circ_b}A\mid S_1\mid H\mid G} $
\end{center}
\begin{center}
$ (\circ_n\Rightarrow) $ $ \dfrac{S_1,A,\neg A\mid H\quad A,S_2\mid G\quad\neg A,S_2\mid G}{{\circ_n}A,S_2\mid S_1\mid H\mid G} $\quad
$ (\Rightarrow\circ_n) $ $ \dfrac{A\Rightarrow \;\mid S_1,A,\neg A\mid H\quad\neg A\Rightarrow\;\mid S_1,A,\neg A\mid G}{S_1,{\circ_n}A\mid H\mid G} $
\end{center}
\begin{center}
$ (\neg\circ_n\Rightarrow) $ $ \dfrac{\Rightarrow A\mid S_1,A,\neg A\mid H\quad\Rightarrow\neg A\mid S_1,A,\neg A\mid G}{\neg{\circ_n}A,S_1\mid H\mid G} $\quad
$ (\Rightarrow\neg\circ_n) $ $ \dfrac{A,\neg A,S_1\mid H\quad A,S_2\mid G\quad\neg A,S_2\mid G}{S_2,\neg{\circ_n}A\mid S_1\mid H\mid G} $
\end{center}
\begin{center}
$ (\circ_f\Rightarrow) $ $ \dfrac{\neg A,S_1, A\mid H\quad S_2,\neg A\mid G\quad A,S_2\mid G}{{\circ_f}A,S_2\mid S_1\mid H\mid G} $\quad
$ (\Rightarrow\circ_f) $ $ \dfrac{A\Rightarrow \mid \neg A,S_1, A\mid H\quad\Rightarrow\neg A\mid \neg A,S_1, A\mid G}{S_1,{\circ_f}A\mid H\mid G} $
\end{center}
\begin{center}
$ (\neg\circ_f\Rightarrow) $ $ \dfrac{A\Rightarrow \mid \neg A,S_1, A\mid H\quad\Rightarrow\neg A\mid \neg A,S_1, A\mid G}{\neg{\circ_f}A,S_1\mid H\mid G} $\quad
$ (\Rightarrow\neg\circ_f) $ $ \dfrac{\neg A,S_1, A\mid H\quad S_2,\neg A\mid G\quad A,S_2\mid G}{S_2,\neg{\circ_f}A\mid S_1\mid H\mid G} $
\end{center}}
\begin{center}

\end{center}
{\footnotesize \begin{center}
$ (\FE_t\Rightarrow) $ $ \dfrac{A,S_1,\neg A\mid H\quad S_2,A\mid G\quad\neg A,S_2\mid G}{{\FE_t}A,S_1\mid S_2\mid H\mid G} $\quad
$ (\Rightarrow\FE_t) $ $ \dfrac{S_1, A\mid A\Rightarrow\neg A\mid H\quad\neg A,S_1\mid A\Rightarrow\neg A\mid G}{S_1,{\FE_t}A\mid H\mid G} $
\end{center}
\begin{center}
$ (\neg\FE_t\Rightarrow) $ $ \dfrac{S_1, A\mid A\Rightarrow\neg A\mid H\quad\neg A,S_1\mid A\Rightarrow\neg A\mid G}{\neg{\FE_t}A,S_1\mid H\mid G} $\quad
%\end{center}
%\begin{center}
$ (\Rightarrow\neg\FE_t) $ $ \dfrac{A,S_1,\neg A\mid H\quad S_2,A\mid G\quad\neg A,S_2\mid G}{S_1,\neg{\FE_t}A\mid S_2\mid H\mid G} $
\end{center}
%\begin{center}

%\end{center}
\begin{center}
$ (\FE_b\Rightarrow) $ $ \dfrac{A,\neg A,S_1\mid H\quad S_2,A\mid G\quad S_2,\neg A\mid G}{{\FE_b}A,S_1\mid S_2\mid H\mid G} $\quad
%\end{center}
%\begin{center}
$ (\Rightarrow\FE_b) $ $ \dfrac{S_1, A\mid A,\neg A\Rightarrow\;\mid H\quad S_1,\neg A\mid A,\neg A\Rightarrow\;\mid G}{S_1,{\FE_b}A\mid H\mid G} $
\end{center}
\begin{center}
$ (\neg\FE_b\Rightarrow) $ $ \dfrac{S_1, A\mid\;\Rightarrow A,\neg A\mid H\quad S_1,\neg A\mid \;\Rightarrow A,\neg A\mid G}{\neg{\FE_b}A,S_1\mid H\mid G} $\quad
%\end{center}
%\begin{center}
$ (\Rightarrow\neg\FE_b) $ $ \dfrac{A,\neg A,S_1\mid H\quad A,S_2\mid G\quad\neg A,S_2\mid G}{S_1,\neg{\FE_b}A\mid S_2\mid H\mid G} $
\end{center}
%\begin{center}

%\end{center}
\begin{center}
$ (\FE_n\Rightarrow) $ $ \dfrac{S_1,A,\neg A\mid H\quad A,S_2\mid G\quad\neg A,S_2\mid G}{{\FE_n}A,S_1\mid S_2\mid H\mid G} $\quad
%\end{center}
%\begin{center}
$ (\Rightarrow\FE_n) $ $ \dfrac{A,S_1\mid\; \Rightarrow A,\neg A\mid H\quad\neg A,S_1\mid\;\Rightarrow A,\neg A\mid G}{S_1,{\FE_n}A\mid H\mid G} $
\end{center}
\begin{center}
$ (\neg\FE_n\Rightarrow) $ $ \dfrac{A,S_1\mid A,\neg A \Rightarrow \;\mid H\quad\neg A,S_1\mid A,\neg A \Rightarrow \;\mid G}{\neg{\FE_n}A,S_1\mid H\mid G} $\quad
%\end{center}
%\begin{center}
$ (\Rightarrow\neg\FE_n) $ $ \dfrac{S_1,A,\neg A\mid H\quad S_2,A\mid G\quad S_2,\neg A\mid G}{S_1,\neg{\FE_n}A\mid S_2\mid H\mid G} $
\end{center}

\begin{center}
$ (\FE_f\Rightarrow) $ $ \dfrac{\neg A,S_1, A\mid H\quad S_2,\neg A\mid G\quad A,S_2\mid G}{{\FE_f}A,S_1\mid S_2\mid H\mid G} $\quad
%\end{center}
%\begin{center}
$ (\Rightarrow\FE_f) $ $ \dfrac{A,S_1 \mid \neg A\Rightarrow A\mid H\quad S_1,\neg A\mid \neg A\Rightarrow A\mid G}{S_1,{\FE_f}A\mid H\mid G} $
\end{center}
\begin{center}
$ (\neg\FE_f\Rightarrow) $ $ \dfrac{A,S_1 \mid \neg A\Rightarrow A\mid H\quad S_1,\neg A\mid \neg A\Rightarrow A\mid G}{\neg{\FE_f}A,S_1\mid H\mid G} $\quad
%\end{center}
%\begin{center}
$ (\Rightarrow\neg\FE_f) $ $ \dfrac{\neg A,S_1, A\mid H\quad S_2,\neg A\mid G\quad A,S_2\mid G}{S_1,\neg{\FE_f}A\mid S_2\mid H\mid G} $
\end{center}}

Hypersequent rules for the second group of Belnapian accident operators are as follows (where $ S_1 $ stands for $ \Gamma\Rightarrow\Delta $ and $ S_2 $ stands for $ \Theta\Rightarrow\Lambda $):
{\footnotesize \begin{center}
$ (\bullet_t\Rightarrow) $ $ \dfrac{\Rightarrow A\mid A, S_1,\neg A\mid H\quad\neg A\Rightarrow\;\mid A, S_1,\neg A\mid G}{\bullet_tA, S_1\mid H\mid G} $\quad
%\end{center}
%\begin{center}
$ (\Rightarrow\bullet_t) $ $ \dfrac{A, S_1,\neg A\mid H\quad S_2,A\mid G\quad\neg A, S_2\mid G}{ S_2,\bullet_tA\mid S_1\mid H\mid G} $
\end{center}
\begin{center}
$ (\neg{\bullet_t}\Rightarrow) $ $ \dfrac{A, S_1,\neg A\mid H\quad S_2,A\mid G\quad\neg A, S_2\mid G}{\neg{\bullet_t}A, S_2\mid S_1\mid H\mid G} $\quad
%\end{center}
%\begin{center}
$ (\Rightarrow\neg{\bullet_t}) $ $ \dfrac{\Rightarrow A\mid A, S_1,\neg A\mid H\quad\neg A\Rightarrow\;\mid A, S_1,\neg A\mid G}{ S_1,\neg{\bullet_t}A\mid H\mid G} $
\end{center}
%\begin{center}

%\end{center}
\begin{center}
$ ({\bullet_b}\Rightarrow) $ $ \dfrac{\Rightarrow A\mid A,\neg A, S_1\mid H\quad\Rightarrow\neg A\mid A,\neg A, S_1\mid G}{{\bullet_b}A, S_1\mid H\mid G} $\quad
%\end{center}
%\begin{center}
$ (\Rightarrow{\bullet_b}) $ $ \dfrac{A,\neg A, S_1\mid H\quad S_2,A\mid G\quad S_2,\neg A\mid G}{ S_2,{\bullet_b}A\mid S_1\mid H\mid G} $
\end{center}
\begin{center}
$ (\neg\bullet_b\Rightarrow) $ $ \dfrac{ S_1,A,\neg A\mid H\quad A, S_2\mid G\quad\neg A, S_2\mid G}{\neg{\bullet_b}A, S_2\mid S_1\mid H\mid G} $\quad
%\end{center}
%\begin{center}
$ (\Rightarrow\neg{\bullet_b}) $ $ \dfrac{A\Rightarrow \;\mid  S_1,A,\neg A\mid H\quad\neg A\Rightarrow\;\mid  S_1,A,\neg A\mid G}{ S_1,\neg{\bullet_b}A\mid H\mid G} $
\end{center}
%\begin{center}

%\end{center}
\begin{center}
$ ({\bullet_n}\Rightarrow) $ $ \dfrac{A\Rightarrow \;\mid  S_1,A,\neg A\mid H\quad\neg A\Rightarrow\;\mid  S_1,A,\neg A\mid G}{{\bullet_n}A, S_1\mid H\mid G} $\quad
%\end{center}
%\begin{center}
$ (\Rightarrow{\bullet_n}) $ $ \dfrac{ S_1,A,\neg A\mid H\quad A, S_2\mid G\quad\neg A, S_2\mid G}{ S_2,{\bullet_n}A\mid S_1\mid H\mid G} $
\end{center}
\begin{center}
$ (\neg{\bullet_n}\Rightarrow) $ $ \dfrac{A,\neg A, S_1\mid H\quad  S_2,A\mid G\quad S_2,\neg A\mid G}{\neg{\bullet_n}A, S_2\mid S_1\mid H\mid G} $\quad
%\end{center}
%\begin{center}
$ (\Rightarrow\neg{\bullet_n}) $ $ \dfrac{\Rightarrow A\mid A,\neg A, S_1\mid H\quad\Rightarrow\neg A\mid A,\neg A,  S_1\mid G}{ S_1,\neg{\bullet_n}A\mid H\mid G} $
\end{center}
%\begin{center}

%\end{center}
\begin{center}
$ ({\bullet_f}\Rightarrow) $ $ \dfrac{A\Rightarrow \mid \neg A, S_1, A\mid H\quad\Rightarrow\neg A\mid \neg A, S_1, A\mid G}{{\bullet_f}A, S_1\mid H\mid G} $\quad
%\end{center}
%\begin{center}
$ (\Rightarrow{\bullet_f}) $ $ \dfrac{\neg A, S_1, A\mid H\quad S_2,\neg A\mid G\quad A, S_2\mid G}{ S_2,{\bullet_f}A\mid S_1\mid H\mid G} $
\end{center}
\begin{center}
$ (\neg{\bullet_f}\Rightarrow) $ $ \dfrac{\neg A, S_1, A\mid H\quad S_2,\neg A\mid G\quad A, S_2\mid G}{\neg{\bullet_f}A, S_2\mid S_1\mid H\mid G} $\quad
%\end{center}
%\begin{center}
$ (\Rightarrow\neg{\bullet_f}) $ $ \dfrac{A\Rightarrow \mid \neg A, S_1, A\mid H\quad\Rightarrow\neg A\mid \neg A, S_1, A\mid G}{ S_1,\neg{\bullet_f}A\mid H\mid G} $
\end{center}}

\begin{center}

\end{center}

{\footnotesize \begin{center}
$ (\FA_t\Rightarrow) $ $ \dfrac{ S_1, A\mid A\Rightarrow\neg A\mid H\quad\neg A, S_1\mid A\Rightarrow\neg A\mid G}{\FA_tA, S_1\mid H\mid G} $\quad
%\end{center}
%\begin{center}
$ (\Rightarrow\FA_t) $ $ \dfrac{A, S_1,\neg A\mid H\quad S_2,A\mid G\quad\neg A, S_2\mid G}{ S_1,\FA_tA\mid S_2\mid H\mid G} $
\end{center}
\begin{center}
$ (\neg{\FA_t}\Rightarrow) $ $ \dfrac{A, S_1,\neg A\mid H\quad S_2,A\mid G\quad\neg A, S_2\mid G}{\neg{\FA_t}A, S_1\mid S_2\mid H\mid G} $\quad
%\end{center}
%\begin{center}
$ (\Rightarrow\neg{\FA_t}) $ $ \dfrac{ S_1, A\mid A\Rightarrow\neg A\mid H\quad\neg A, S_1\mid A\Rightarrow\neg A\mid G}{ S_1,\neg{\FA_t}A\mid H\mid G} $
\end{center}
%\begin{center}

%\end{center}
\begin{center}
$ ({\FA_b}\Rightarrow) $ $ \dfrac{ S_1, A\mid A,\neg A\Rightarrow\;\mid H\quad S_1,\neg A\mid A,\neg A\Rightarrow\;\mid G}{{\FA_b}A, S_1\mid H\mid G} $\quad
%\end{center}
%\begin{center}
$ (\Rightarrow{\FA_b}) $ $ \dfrac{A,\neg A, S_1\mid H\quad S_2,A\mid G\quad S_2,\neg A\mid G}{ S_1,{\FA_b}A\mid S_2\mid H\mid G} $
\end{center}
\begin{center}
$ (\neg\FA_b\Rightarrow) $ $ \dfrac{ S_1,A,\neg A\mid H\quad A, S_2\mid G\quad\neg A, S_2\mid G}{\neg{\FA_b}A, S_1\mid S_2\mid H\mid G} $\quad
%\end{center}
%\begin{center}
$ (\Rightarrow\neg{\FA_b}) $ $ \dfrac{A, S_1\mid \;\Rightarrow A,\neg A\mid H\quad\neg A, S_1\mid \;\Rightarrow A,\neg A\mid G}{ S_1,\neg{\FA_b}A\mid H\mid G} $
\end{center}
%\begin{center}

%\end{center}
\begin{center}
$ ({\FA_n}\Rightarrow) $ $ \dfrac{A, S_1\mid\; \Rightarrow A,\neg A\mid H\quad\neg A, S_1\mid\;\Rightarrow A,\neg A\mid G}{{\FA_n}A, S_1\mid H\mid G} $\quad
%\end{center}
%\begin{center}
$ (\Rightarrow{\FA_n}) $ $ \dfrac{ S_1,A,\neg A\mid H\quad A, S_2\mid G\quad\neg A, S_2\mid G}{ S_1,{\FA_n}A\mid S_2\mid H\mid G} $
\end{center}
\begin{center}
$ (\neg{\FA_n}\Rightarrow) $ $ \dfrac{A,\neg A, S_1\mid H\quad  S_2,A\mid G\quad S_2,\neg A\mid G}{\neg{\FA_n}A, S_1\mid S_2\mid H\mid G} $\quad
%\end{center}
%\begin{center}
$ (\Rightarrow\neg{\FA_n}) $ $ \dfrac{ S_1,A\mid A,\neg A \Rightarrow \;\mid H\quad S_1,\neg A\mid A,\neg A\Rightarrow \;\mid G}{ S_1,\neg{\FA_n}A\mid H\mid G} $
\end{center}
%\begin{center}

%\end{center}
\begin{center}
$ ({\FA_f}\Rightarrow) $ $ \dfrac{A, S_1 \mid \neg A\Rightarrow A\mid H\quad S_1,\neg A\mid \neg A\Rightarrow A\mid G}{{\FA_f}A, S_1\mid H\mid G} $\quad
%\end{center}
%\begin{center}
$ (\Rightarrow{\FA_f}) $ $ \dfrac{\neg A, S_1, A\mid H\quad S_2,\neg A\mid G\quad A, S_2\mid G}{ S_1,{\FA_f}A\mid S_2\mid H\mid G} $
\end{center}
\begin{center}
$ (\neg{\FA_f}\Rightarrow) $ $ \dfrac{\neg A, S_1, A\mid H\quad S_2,\neg A\mid G\quad A, S_2\mid G}{\neg{\FA_f}A, S_1\mid S_2\mid H\mid G} $\quad
%\end{center}
%\begin{center}
$ (\Rightarrow\neg{\FA_f}) $ $ \dfrac{A, S_1 \mid \neg A\Rightarrow A\mid H\quad S_1,\neg A\mid \neg A\Rightarrow A\mid G}{ S_1,\neg{\FA_f}A\mid H\mid G} $
\end{center}}
\end{definition}
\section{Embedding theorems}\label{Embedding}
\begin{definition}\label{f}	The mapping $ \f $ from the language $\Lang_{\Box\Diamond}^\neg$ into the language $\Lang_{\Box\Diamond}$ is inductively defined as follows\textup{:}
{\small \begin{enumerate}[$ - $]
\item $ \f(p)=p $ and $ \f(\neg p):=p^\prime $ (where $p^\prime\in\mathcal{P}^\prime $), for each $p\in\mathcal{P}$,
\vspace{-.15in}\begin{multicols}{3}
\item $ \f(\non A)=\non \f(A) $,
\item $ \f(\neg\non A)= \non \f(\neg A) $,
\item $ \f(\neg\neg A)= \f(A) $,
\end{multicols}\vspace{-.15in} \vspace{-.15in}
\begin{multicols}{2}\itemsep=0pt
\item $ \f(\Box A)=\Box \f(A), $
\item $ \f(\D A)=\D \f(A), $
\item $ \f(A\wedge B)=\f(A)\wedge \f(B), $
\item $ \f(A\vee B)=\f(A)\vee \f(B), $
\item $ \f(A\rightarrow B)=\f(A)\rightarrow \f(B), $  
\item $ \f(\neg\Box A)=\D \f(\neg A), $
\item $ \f(\neg\D A)=\Box \f(\neg A), $
\item $ \f(\neg(A\wedge B))=\f(\neg A)\vee \f(\neg B), $
\item $ \f(\neg(A\vee B))=\f(\neg A)\wedge \f(\neg B), $
\item $ \f(\neg(A\rightarrow B))=\f(A)\wedge \f(\neg B). $
\end{multicols}\vspace{-.15in}
\end{enumerate}}
\end{definition}

A similar translation was used in \cite{Kamide19,Kamide23,MML,MMLS5}. We write $ \f(\Gamma) $ for $ \{\f(A)\mid A\in\Gamma \} $ and $ \f(H) $ for $ \f(\Gamma_1)\Rightarrow \f(\Delta_1)\mid \ldots \mid \f(\Gamma_n)\Rightarrow \f(\Delta_n) $, where $ H=\Gamma_1\Rightarrow \Delta_1\mid \ldots \mid \Gamma_n\Rightarrow \Delta_n $. Let $\mathbf{H}\setminus \mathrm{(Cut)}$ be a cut-free fragment of a hypersequent calculus $ \mathbf{H} $ obtained by deleting (Cut).

\begin{lemma}[Weak Syntactical Embedding from $ \bf FDE^{\bf S5}_{\Box\D}$ to $ \bf S5$]\label{WSynt4to2}
For each hypersequent $ \mathscr{H} $, it holds that 
\begin{enumerate}[$ (1) $]\itemsep=0pt
\item $ {\bf FDE^{\bf S5}_{\B\D}}\vdash\mathscr{H}$ implies $ {\bf S5}\vdash\f(\mathscr{H})$,
\item $ {\bf S5}\setminus \mathrm{(Cut)}\vdash\f(\mathscr{H})$ implies $ {\bf FDE^{\bf S5}_{\B\D}}\setminus \mathrm{(Cut)}\vdash\mathscr{H}$.
\end{enumerate}
\end{lemma}
\begin{proof}
(1) By induction on the proof $ \mathfrak{P}_1 $ of $ \mathscr{H} $ in the hypersequent calculus for $ \bf FDE^{\bf S5}_{\B\D}$. We distinguish the cases according to the
last inference of $ \mathfrak{P}_1 $. 

The case $ (\neg\D\Rightarrow) $. The last inference of $ \mathfrak{P}_1 $ has the following form:

\begin{center}
\EnableBpAbbreviations
\AXC{$\neg A,\Gamma\Rightarrow \Delta\mid H$}
\RL{$ (\neg\D \Rightarrow) $}
\UIC{$ \neg\D A\Rightarrow\; \mid \Gamma\Rightarrow \Delta\mid H$}
\DisplayProof
\end{center}

Using the induction hypothesis, we obtain the following proof in the hypersequent calculus for $ \bf S5$, where Df $ \f $ stands for Definition \ref{f}:
\begin{center}
\EnableBpAbbreviations
\AXC{$ \f(\neg A),\f(\Gamma)\Rightarrow \f(\Delta)\mid \f(H)$}
\RL{$ (\B \Rightarrow) $}
\UIC{$ \Box\f(\neg A)\Rightarrow\; \mid \f(\Gamma)\Rightarrow \f(\Delta)\mid \f(H)$}
\RL{Df $ \f $}
\UIC{$ \f(\neg\D A)\Rightarrow\; \mid \f(\Gamma)\Rightarrow \f(\Delta)\mid \f(H)$}
\DisplayProof
\end{center}

The other cases are treated similarly.

(2) By induction on the proof $ \mathfrak{P}_2 $ of $\f(\mathscr{H} )$ in ${\bf S5}\setminus \mathrm{(Cut)}$. Similarly to (1).
\end{proof}

\begin{theorem}[Syntactical embedding from $\bf FDE^{\bf S5}_{\Box\D}$ into \textbf{S5}]\label{Synt4to2} 
For each  hypersequent $ \mathscr{H} $, 
\begin{enumerate}[$(1)$]\itemsep=0pt
   \item  ${\bf FDE^{\bf S5}_{\Box\D}}\vdash \mathscr{H}$ iff ${\bf S5}\vdash \f(\mathscr{H})$,     \item  ${\bf FDE^{\bf S5}_{\Box\D}}\setminus \mathrm{(Cut)}\vdash \mathscr{H}$ iff ${\bf S5}\setminus \mathrm{(Cut)}\vdash \f(\mathscr{H})$.
   \end{enumerate}
\end{theorem}
\begin{proof}
(1) By Lemma \ref{WSynt4to2}, ${\bf FDE^{\bf S5}_{\B\D}}\vdash\mathscr{H}$ implies $ {\bf S5}\vdash\f(\mathscr{H})$. For another direction suppose that ${\bf S5}\vdash \f(\mathscr{H})$. By the cut admissibility theorem for Restall's hypersequent calculus\footnote{Restall \cite{Restall} gives a completeness proof for his calculus and uses a technique which ``can succeed without the use of the cut rule''. \cite[p. 166]{Restall}. He also presents a sketch of a constructive cut elimination proof. A more detailed version of such a proof is due to Indrzejczak \cite{IndrzejczakSC}. In \cite{PetrukhinS5}, it is pointed out that the rule (EW) is required for proper constructive cut elimination proof.} for \textbf{S5}, ${\bf S5}\setminus(\mathrm{Cut})\vdash \f(\mathscr{H})$. By Lemma \ref{WSynt4to2}, $ {\bf FDE^{\bf S5}_{\B\D}}\setminus \mathrm{(Cut)}\vdash\mathscr{H}$. Therefore, $ {\bf FDE^{\bf S5}_{\B\D}}\vdash\mathscr{H}$.
				
(2) The proof is similar.
\end{proof}

\begin{theorem}[Cut admissibility]\label{CutAdm1}
The rule (Cut) is admissible in $\bf FDE^{\bf S5}_{\Box\D}\setminus \mathrm{(Cut)}$.
\end{theorem}
\begin{proof}
Follows from Theorem \ref{Synt4to2} and the cut admissibility theorem for Restall's hypersequent calculus for \textbf{S5}.
\end{proof}
\begin{theorem}[Semantic embedding from $\bf FDE^{\bf S5}_{\Box\D}$ into \textbf{S5}]\label{Sem4to2} 
For every $\Lang_{\Box\Diamond}^\neg$-formula $ A $, it holds that ${\bf FDE^{\bf S5}_{\Box\D}}\models A$ iff ${\bf S5}\models \f(A)$.
\end{theorem}
\begin{proof}
This theorem follows from the following Lemmas \ref{Sem4to2L1} and \ref{Sem4to2L2}.
\end{proof}

\begin{lemma}\label{Sem4to2L1}
For any $ {\bf FDE^{\bf S5}_{\B\D}} $-model $ \mathfrak{M}_1=\langle W,\var\rangle $ we can construct an $ {\bf S5} $-model $ \mathfrak{M}_2=\langle W,\var^\prime\rangle $ such that for each $\Lang_{\B\D}^\neg$-formula $ \A $ and any $ x\in W$, it is true that  
\begin{enumerate}[$\blacktriangleright$]\itemsep=0pt
\item 
$ \one\in\var(A,x) $ in $ \mathfrak{M}_1 $ iff $ \var^\prime(\f(A),x)=\one $ in $ \mathfrak{M}_2 $,
\item 
$ \one\in\var(\neg A,x) $ in $ \mathfrak{M}_1 $ iff $ \var^\prime(\f(\neg A),x)=\one $ in $ \mathfrak{M}_2 $.
\end{enumerate}
\end{lemma}
\begin{proof}
Assume that $ \mathfrak{M}_1=\langle W,\var\rangle $ is $ {\bf FDE^{\bf S5}_{\B\D}} $-model and $ \mathfrak{M}_2=\langle W,\var^\prime\rangle $ is $ {\bf S5} $-model, where $ \var^\prime $ is a valuation such that for all $ p \in \Prop$ and $ x\in W$:	
\begin{enumerate}[$\blacktriangleright$]\itemsep=0pt
\item $ \one\in\var(p,x) $ in $ \mathfrak{M}_1 $ iff $ \var^\prime(p,x)=\one $ in $ \mathfrak{M}_2 $,
\item $ \one\in\var(\neg p,x) $ in $ \mathfrak{M}_1 $ iff $ \var^\prime(p^\prime,x)=\one $ in $ \mathfrak{M}_2 $.
	\end{enumerate}
	
The proof is by induction on the construction of the $\Lang_{\B\D}^\neg$-formula $ \A $.
	
	\begin{enumerate}[$ (1) $]\itemsep=0pt
\item $ \A $ is $ p $ or $ \neg p $. Follows from the assumption and Definition \ref{f}.
		
\item $ \A $ is $ \Be\rightarrow\Ce $. By Definition \ref{KripkeNelsonval}, $ \one\in\var(\Be\rightarrow\Ce,x) $ in $ \mathfrak{M}_1$ iff $ \one\in\var(\Be,x) $ in $ \mathfrak{M}_1$ implies $ \one\in\var(\Ce,x) $ in $ \mathfrak{M}_1$ iff (by the induction hypothesis) $ \var^\prime(\f(\Be),x)=\one $ in $ \mathfrak{M}_2$ implies $ \var^\prime(\f(\Ce),x)=\one $ in $ \mathfrak{M}_2$ iff $ \var^\prime(\f(\Be)\rightarrow \f(\Ce),x)=\one $ in $ \mathfrak{M}_2$ iff (Definition \ref{f}) $ \var^\prime(\f(\Be\rightarrow\Ce),x)=\one $ in $ \mathfrak{M}_2$. 

\item $ \A $ is $ \neg(\Be\rightarrow\Ce) $. 
 By Definition \ref{KripkeNelsonval}, $ \one\in\var(\neg(\Be\rightarrow\Ce),x) $ in $ \mathfrak{M}_1$ iff $ \zero\in\var(\Be\rightarrow\Ce,x) $ in $ \mathfrak{M}_1$ iff $ \one\in\var(\Be,x) $ in $ \mathfrak{M}_1$ and $ \zero\in\var(C,x) $ in $ \mathfrak{M}_1$ iff $ \one\in\var(\Be,x) $ in $ \mathfrak{M}_1$ and $ \one\in\var(\neg C,x) $ in $ \mathfrak{M}_1$ iff (by the induction hypothesis) $ \var^\prime(\f(\Be),x)=\one $ in $ \mathfrak{M}_2$ and $ \var^\prime(\f(\neg C),x)=\one $ in $ \mathfrak{M}_2$ iff $ \var^\prime(\f(\Be)\wedge \f(\neg C),x)=\one $ in $ \mathfrak{M}_2$ iff (Definition \ref{f})  $ \var^\prime(\f(\neg(\Be\rightarrow C)),x)=\one $ in $ \mathfrak{M}_2$.

\item $ A $ is $ \neg\Box B$. By Definition \ref{KripkeNelsonval}, $ \one\in\var(\neg\Box B,x) $ in $ \mathfrak{M}_1$ iff $ \zero\in\var(\Box\Be,x) $ in $ \mathfrak{M}_1$ iff  
$\exists_{y\in W} \zero\in\var(\Be,y) $ in $ \mathfrak{M}_1$ iff  
$\exists_{y\in W} \one\in\var(\neg\Be,y) $ in $ \mathfrak{M}_1$ 
iff (by the induction hypothesis) $\exists_{y\in W}\var^\prime(\f(\neg\Be),y)=\one$ in $ \mathfrak{M}_2$ iff (by Definition \ref{Kripkeval}) $\var^\prime(\D \f(\neg\Be),x)=\one$ in $ \mathfrak{M}_2$ iff (by Definition \ref{f}) $\var^\prime(\f(\neg\Box\Be),x)=\one$ in~$ \mathfrak{M}_2$.
	\end{enumerate}	
	The other cases are treated similarly.
\end{proof} 
\begin{lemma}\label{Sem4to2L2}
For any $ {\bf S5} $-model $ \mathfrak{M}_1=\langle W,\var\rangle $ we can construct an $ {\bf FDE^{\bf S5}_{\B\D}} $-model $ \mathfrak{M}_2=\langle W,\var^\prime\rangle $ such that for each $\Lang_{\B\D}^\neg$-formula $ \A $ and any $ x\in W$, it is true that 
\begin{enumerate}[$\blacktriangleright$]\itemsep=0pt
\item 
$ \var(\f(A),x)=\one $ in $ \mathfrak{M}_1 $ iff $ \one\in\var^\prime(A,x) $ in $ \mathfrak{M}_2 $,
\item 
$ \var(\f(\neg A),x)=\one $ in $ \mathfrak{M}_1 $ iff $ \one\in\var^\prime(\neg A,x) $ in $ \mathfrak{M}_2 $.
\end{enumerate}
\end{lemma}
\begin{proof}
	Similarly to Lemma \ref{Sem4to2L1}.
\end{proof}

\begin{theorem}[Soundness and completeness of $ {\bf FDE^{\bf S5}_{\Box\D}} $]\label{FDES5SoundCompl}
For every $\Lang_{\Box\Diamond}^\neg$-formula $ A $, it holds that ${\bf FDE^{\bf S5}_{\Box\D}}\models A$ iff ${\bf FDE^{\bf S5}_{\Box\D}}\vdash A$.
\end{theorem}
\begin{proof}
Follows from Theorems \ref{Synt4to2} and \ref{Sem4to2}.
\end{proof}

\begin{definition}\label{g} The mapping $ \g $ from the language $\Lang_{\circ\bullet}^\neg$ into the language $\Lang_{\Box\Diamond}^\neg$ is inductively defined as follows, where $ I\in\{T,F\} $ and $ i\in\{t,b,n,f \} $\textup{:}
{\small \begin{enumerate}[$ - $]\itemsep=0pt
\item $\g(p)= p$, for each $p\in\mathcal{P}$,
\item $\g(\diamond A)=\diamond\g(A) $, where $ \diamond\in\{\neg,\non \} $,
\item $\g(A\odot B)=\g(A)\odot\g(B) $, where $ \odot\in\{\wedge,\vee,\rightarrow \} $,
\vspace{-.15in}\begin{multicols}{2}
\item $\g(\circ_I A)=\g(A^I)\rightarrow\Box\g(A^I) $,
\item $\g(\FE_I A)=\non\g(A^I)\rightarrow\Box \non\g(A^I) $,
\item $\g(\bullet_I A)=\g(A^I)\wedge\D\non\g(A^I) $,
\item $\g(\FA_I A)=\non\g(A^I)\wedge\D\g(A^I) $,
\item $\g({\circ_i} A)=\g(A^i)\rightarrow\Box\g(A^i)$,
\item $\g({\FE_i} A)=\non\g(A^i)\rightarrow\Box\non\g(A^i)$,
\item $\g({\bullet_i} A)=\g(A^i)\wedge\D\non\g(A^i)$,
\item $\g({\FA_i} A)=\non\g(A^i)\wedge\D\g(A^i)$,
\end{multicols}\vspace{-.15in}
\end{enumerate}}
\end{definition} \vspace{-.15in}

\begin{lemma}[Weak Syntactical Embedding from $ \bf FDE^{\bf S5}_{\circ\bullet}$ to $ \bf FDE^{\bf S5}_{\B\D}$]\label{WSynt4to4}
For each hypersequent $ \mathscr{H} $, it holds that 
\begin{enumerate}[$ (1) $]\itemsep=0pt
\item $ {\bf FDE^{\bf S5}_{\circ\bullet}}\vdash\mathscr{H}$ implies $ {\bf FDE^{\bf S5}_{\Box\D}}\vdash\g(\mathscr{H})$,
\item $ {\bf FDE^{\bf S5}_{\Box\D}}\setminus \mathrm{(Cut)}\vdash\g(\mathscr{H})$ implies $ {\bf FDE^{\bf S5}_{\circ\bullet}}\setminus \mathrm{(Cut)}\vdash\mathscr{H}$.
\end{enumerate}
\end{lemma}
\begin{proof}
(1) By induction on the proof $ \mathfrak{P}_1 $ of $ \mathscr{H} $ in the hypersequent calculus for $ \bf FDE^{\bf S5}_{\circ\bullet}$. We distinguish the cases according to the
last inference of $ \mathfrak{P}_1 $. 

\begin{enumerate}
\item The case $ (\circ_T\Rightarrow) $. The last inference of $ \mathfrak{P}_1 $ has the following form:

\begin{center}
\EnableBpAbbreviations
\AXC{$ A,\Gamma\Rightarrow \Delta\mid H$}
\AXC{$\Theta\Rightarrow \Lambda, A\mid G $}
\RL{$ (\circ_T \Rightarrow) $}
\BIC{$ \circ_TA,\Theta\Rightarrow \Lambda\mid \Gamma\Rightarrow \Delta\mid H\mid G$}
\DisplayProof
\end{center}

Using the induction hypothesis, we obtain the following proof in the hypersequent calculus for $ \bf FDE^{\bf S5}_{\Box\D}$, where Df $ \g $ stands for Definition \ref{g}:
\begin{center}
\EnableBpAbbreviations
\AXC{$\g(\Theta)\Rightarrow\g(\Lambda),\g(A)\mid\g(G) $}
\AXC{$\g(A),\g(\Gamma)\Rightarrow\g(\Delta)\mid\g(H) $}
\RL{$ (\Box \Rightarrow) $}
\UIC{$\Box\g(A)\Rightarrow\; \mid\g(\Gamma)\Rightarrow\g(\Delta)\mid\g(H) $}
\RL{$ (\rightarrow \Rightarrow) $}
\BIC{$\g(A)\rightarrow \Box\g(A),\g(\Theta)\Rightarrow\g(\Lambda)\mid\g(\Gamma)\Rightarrow\g(\Delta)\mid\g(H)\mid\g(G)$}
\RL{Df $ \g $ %$ \circ_T $
}
\UIC{$\g(\circ_TA),\g(\Theta)\Rightarrow\g(\Lambda)\mid\g(\Gamma)\Rightarrow\g(\Delta)\mid\g(H)\mid\g(G)$}
\DisplayProof
\end{center}

\item The case $ (\FE_T\Rightarrow) $. The last inference of $ \mathfrak{P}_1 $ has the following form:

\begin{center}
\EnableBpAbbreviations
\AXC{$A,\Gamma\Rightarrow \Delta\mid H$}
\AXC{$\Theta\Rightarrow \Lambda,A\mid G $}
\RL{$ (\FE_T \Rightarrow) $}
\BIC{$ \FE_TA,\Gamma\Rightarrow \Delta\mid \Theta\Rightarrow \Lambda\mid H\mid G$}
\DisplayProof
\end{center}

Using the induction hypothesis, we obtain the following proof in $ \bf FDE^{\bf S5}_{\Box\D}$:
\begin{center}
\EnableBpAbbreviations
\AXC{$\g(A),\g(\Gamma)\Rightarrow\g(\Delta)\mid\g(H) $}
\RL{$ ( \Rightarrow\non) $}
\UIC{$\g(\Gamma)\Rightarrow\g(\Delta),\non\g(A)\mid\g(H) $}
\AXC{$\g(\Theta)\Rightarrow\g(\Lambda),\g(A)\mid\g(G) $}
\RL{$ (\non \Rightarrow) $}
\UIC{$\non\g(A),\g(\Theta)\Rightarrow\g(\Lambda)\mid\g(G) $}
\RL{$ (\Box \Rightarrow) $}
\UIC{$\Box\non\g(A)\Rightarrow\; \mid\g(\Theta)\Rightarrow\g(\Lambda)\mid\g(H) $}
\RL{$ (\rightarrow \Rightarrow) $}
\BIC{$ \non\g(A)\rightarrow \Box\non\g(A),\g(\Gamma)\Rightarrow\g(\Delta)\mid\g(\Theta)\Rightarrow\g(\Lambda)\mid\g(H)\mid\g(G)$}
\RL{Df $ \g $}
\UIC{$\g(\FE_TA),\g(\Gamma)\Rightarrow\g(\Delta)\mid\g(\Theta)\Rightarrow\g(\Lambda)\mid\g(H)\mid\g(G)$}
\DisplayProof
\end{center}

\item 	The case $ (\Rightarrow\FE_T) $. The last inference of $ \mathfrak{P}_1 $ has the following form:

\begin{center}
 $ \dfrac{\G\Rightarrow\Delta, A\mid A\Rightarrow \mid H}{\G\Rightarrow \Delta,\FE_TA\mid H} $  $ (\Rightarrow\FE_T) $
\end{center}

Using the induction hypothesis, we obtain the following proof in $ \bf FDE^{\bf S5}_{\Box\D}$:
\begin{center}
\EnableBpAbbreviations
\AXC{$\g(\G)\Rightarrow\g(\Delta),\g(A)\mid\g(A)\Rightarrow\;\mid\g(H) $}
\RL{$ ( \Rightarrow\non) $}
\UIC{$\g(\G)\Rightarrow\g(\Delta),\g(A)\mid\; \Rightarrow\non\g(A)\mid\g(H) $}
\RL{$ (\non \Rightarrow) $}
\UIC{$ \non\g(A),\g(\G)\Rightarrow\g(\Delta)\mid\;\Rightarrow\non\g(A)\mid\g(H) $}
\RL{$ ( \Rightarrow\Box) $}
\UIC{$ \non\g(A),\g(\G)\Rightarrow\g(\Delta)\mid\;\Rightarrow\Box\non\g(A)\mid\g(H) $}
\RL{(Merge)}
\UIC{$\non\g(A),\g(\G)\Rightarrow\g(\Delta),\Box\non\g(A)\mid\g(H) $}
\RL{$ (\Rightarrow\rightarrow) $}
\UIC{$\g(\G)\Rightarrow\g(\Delta),\non\g(A)\rightarrow\Box\non\g(A)\mid\g(H)$}
\RL{Df $ \g $ %$ \FE_T $
}
\UIC{$\g(\G)\Rightarrow\g(\Delta),\g(\FE_TA)\mid\g(H)$}
\DisplayProof
\end{center}

\item 	The case $ (\Rightarrow\neg\FE_F) $. The last inference of $ \mathfrak{P}_1 $ has the following form:

\begin{center}
 $ \dfrac{\neg A,\G\Rightarrow\Delta\mid H\qquad A,\Theta\Rightarrow\Lambda \mid G}{\G\Rightarrow \Delta,\neg{\FE}_FA\mid\Theta\Rightarrow\Lambda\mid H\mid G} $  $ (\Rightarrow\neg\FE_F) $
\end{center}

Using the induction hypothesis, we obtain the following proof in $ \bf FDE^{\bf S5}_{\Box\D}$:
{\small \begin{center}
\EnableBpAbbreviations
\AXC{$\g(\neg A),\g(\G)\Rightarrow\g(\Delta)\mid \g(H) $}
\RL{$ ( {\Rightarrow}\non) $}
\UIC{$\g(\G)\Rightarrow\g(\Delta), \non\g(\neg A)\mid\g(H) $}
\AXC{$\g(A),\g(\Theta)\Rightarrow\g(\Lambda)\mid\g(G) $}
\RL{$ (\neg\neg{\Rightarrow}) $}
\UIC{$\neg\neg\g(A),\g(\Theta)\Rightarrow\g(\Lambda)\mid\g(G) $}
\RL{$ ( {\Rightarrow}\neg\non) $}
\UIC{$\g(\Theta)\Rightarrow\g(\Lambda),\neg\non\neg\g(A)\mid\g(G) $}
\RL{$ ( {\Rightarrow}\neg\Box) $}
\UIC{$\g(\Theta)\Rightarrow\g(\Lambda)\mid\;\Rightarrow\neg\Box\non\neg\g(A)\mid\g(G)$}
\RL{Df $ \g $}
\UIC{$\g(\Theta)\Rightarrow\g(\Lambda)\mid\;\Rightarrow\neg\Box\non\g(\neg A)\mid\g(G)$}
\RL{$ ( {\Rightarrow}\neg{\rightarrow}) $}
\BIC{$\g(\G)\Rightarrow\g(\Delta),\neg(\non\g(\neg A)\rightarrow\Box\non\g(\neg A)) \mid\g(\Theta)\Rightarrow\g(\Lambda)\mid\g(H)\mid\g(G) $}
\RL{Df $ \g $}
\UIC{$\g(\G)\Rightarrow\g(\Delta),\g(\neg\FE_FA)\mid\g(\Theta)\Rightarrow\g(\Lambda)\mid\g(H)\mid\g(G)$}
\DisplayProof
\end{center}}

\item 	The case $ (\neg\circ_T\Rightarrow) $. The last inference of $ \mathfrak{P}_1 $ has the following form:

\begin{center}
 $ \dfrac{\neg A\Rightarrow\;\mid A,\G\Rightarrow\Delta\mid H}{\neg{\circ_T}A,\G\Rightarrow \Delta\mid H} $  $ (\neg\circ_T\Rightarrow) $
\end{center}

Using the induction hypothesis, we obtain the following proof in $ \bf FDE^{\bf S5}_{\Box\D}$:
\begin{center}
\EnableBpAbbreviations
\AXC{$\g(\neg A)\Rightarrow\;\mid\g(A),\g(\G)\Rightarrow\g(\Delta)\mid \g(H) $}
\RL{Df $\g$}
\UIC{$\neg\g(A)\Rightarrow\;\mid\g(A),\g(\G)\Rightarrow\g(\Delta)\mid \g(H) $}
\RL{$ (\neg\Box\Rightarrow) $}
\UIC{$\neg\Box\g(A)\Rightarrow\;\mid\g(A),\g(\G)\Rightarrow\g(\Delta)\mid \g(H)$}
\RL{(Merge)}
\UIC{$\g(A),\neg\Box\g(A),\g(\G)\Rightarrow\g(\Delta)\mid \g(H)$}
\RL{$ (\neg{\rightarrow}\Rightarrow) $}
\UIC{$\neg(\g(A)\rightarrow\Box\g(A)),\g(\G)\Rightarrow\g(\Delta)\mid \g(H)$}
\RL{Df $ \g $}
\UIC{$\g(\neg{\circ_T}A),\g(\G)\Rightarrow\g(\Delta)\mid \g(H)$}
\DisplayProof
\end{center}

\item The case $ (\circ_t\Rightarrow) $. The last inference of $ \mathfrak{P}_1 $ has the following form:

\begin{center}
 $ \dfrac{A,\Gamma\Rightarrow\Delta,\neg A\mid H\quad\Theta\Rightarrow\Lambda,A\mid G\quad\neg A,\Theta\Rightarrow\Lambda\mid G}{{\circ_t}A,\Theta\Rightarrow\Lambda\mid\Gamma\Rightarrow\Delta\mid H\mid G} $  $ (\circ_t\Rightarrow) $
\end{center}

Using the induction hypothesis, we obtain the following proof in $ \bf FDE^{\bf S5}_{\Box\D}$ (due to space limits we assume that $ H=G=\emptyset $ and omit names of the rules, double lines indicate several applications of the rules):
{\small \begin{center}
\EnableBpAbbreviations
\AXC{$\g(\Theta)\Rightarrow\g(\Lambda),\g(A) $}
\AXC{$\g(\neg A),\g(\Theta)\Rightarrow\g(\Lambda) $}
\UIC{$\g(\Theta)\Rightarrow\g(\Lambda),\non\g(\neg A) $}
\BIC{$\g(\Theta),\g(\Theta)\Rightarrow\g(\Lambda),\g(\Lambda),\g(A)\wedge\non\g(\neg A) $}
\doubleLine
\UIC{$\g(\Theta)\Rightarrow\g(\Lambda),\g(A)\wedge\non\g(\neg A)$}

\AXC{$\g(A),\g(\Gamma)\Rightarrow\g(\Delta),\g(\neg A) $}
\UIC{$\g(A),\non\g(\neg A),\g(\Gamma)\Rightarrow\g(\Delta) $}
\UIC{$\g(A)\wedge\non\g(\neg A),\g(\Gamma)\Rightarrow\g(\Delta) $}
\UIC{$\Box(\g(A)\wedge\non\g(\neg A))\Rightarrow \mid\g(\Gamma)\Rightarrow\g(\Delta) $}
\BIC{$(\g(A)\wedge\non\g(\neg A))\rightarrow\Box(\g(A)\wedge\non\g(\neg A)),\g(\Theta)\Rightarrow\g(\Lambda)\mid\g(\Gamma)\Rightarrow\g(\Delta)$}
\RL{Df $ \g $}
\UIC{$\g(\circ_tA),\g(\Theta)\Rightarrow\g(\Lambda)\mid\g(\Gamma)\Rightarrow\g(\Delta)$}
\DisplayProof
\end{center}}

\item The case $ (\Rightarrow\circ_t) $. The last inference of $ \mathfrak{P}_1 $ has the following form:

\begin{center}
 $ \dfrac{\Rightarrow A\mid A,\Gamma\Rightarrow\Delta,\neg A\mid H\quad\neg A\Rightarrow\;\mid A,\Gamma\Rightarrow\Delta,\neg A\mid G}{\Gamma\Rightarrow\Delta,{\circ_t}A\mid H\mid G} $  $ (\Rightarrow\circ_t) $
\end{center}

Using the induction hypothesis, we obtain the following proof in $ \bf FDE^{\bf S5}_{\Box\D}$ (due to space limits we assume that $ H=G=\emptyset $ and omit some names of the rules):
{\small \begin{center}
\EnableBpAbbreviations
\AXC{$\Rightarrow\g(A)\mid\g(A),\g(\Gamma)\Rightarrow\g(\Delta),\g(\neg A) $}
%\RL{$ (\non{\Rightarrow}) $}
\UIC{$\Rightarrow\g(A)\mid\g(A),\non\g(\neg A),\g(\Gamma)\Rightarrow\g(\Delta)$}
%\RL{$ (\wedge{\Rightarrow}) $}
\UIC{$\Rightarrow\g(A)\mid\g(A)\wedge\non\g(\neg A),\g(\Gamma)\Rightarrow\g(\Delta)$}

\AXC{$\g(\neg A)\Rightarrow\;\mid\g(A),\g(\Gamma)\Rightarrow\g(\Delta),\g(\neg A)$}
%\LL{$ (\non{\Rightarrow}) $}
\UIC{$\g(\neg A)\Rightarrow\;\mid\g(A),\non\g(\neg A),\g(\Gamma)\Rightarrow\g(\Delta)$}
%\LL{$ (\wedge{\Rightarrow}) $}
\UIC{$\g(\neg A)\Rightarrow\;\mid\g(A)\wedge\non\g(\neg A),\g(\Gamma)\Rightarrow\g(\Delta)$}
%\LL{$ ({\Rightarrow}\non) $}
\UIC{$\Rightarrow\non\g(\neg A)\mid\g(A)\wedge\non\g(\neg A),\g(\Gamma)\Rightarrow\g(\Delta)$}

%\LL{$ ({\Rightarrow}\wedge) $}
\BIC{$\Rightarrow\g(A)\wedge\non\g(\neg A)\mid\g(A)\wedge\non\g(\neg A),\g(\Gamma)\Rightarrow\g(\Delta)$}
%\LL{$ ({\Rightarrow}\wedge)\qquad\qquad $}
\RL{$ (\Box \Rightarrow) $}
\UIC{$\Rightarrow \Box(\g(A)\wedge\non\g(\neg A))\mid\g(A)\wedge\non\g(\neg A),\g(\Gamma)\Rightarrow\g(\Delta)$}
\RL{(Merge)}
\UIC{$\g(A)\wedge\non\g(\neg A),\g(\Gamma)\Rightarrow\g(\Delta),\Box(\g(A)\wedge\non\g(\neg A))$}
\RL{$ (\Rightarrow\rightarrow) $}
\UIC{$\g(\Gamma)\Rightarrow\g(\Delta),(\g(A)\wedge\non\g(\neg A))\rightarrow\Box(\g(A)\wedge\non\g(\neg A))$}
\RL{Df $ \g $}
\UIC{$\g(\Gamma)\Rightarrow\g(\Delta),\g(\circ_tA)$}
\DisplayProof
\end{center}}

\item The case $ (\Rightarrow\neg\circ_b) $. The last inference of $ \mathfrak{P}_1 $ has the following form:

\begin{center}
 $ \dfrac{\Gamma\Rightarrow\Delta,A,\neg A\mid H\quad\Theta\Rightarrow\Lambda,A\mid G\quad \Theta\Rightarrow\Lambda,\neg A\mid G}{\Theta\Rightarrow\Lambda,\neg{\circ_b}A\mid \Gamma\Rightarrow\Delta\mid H\mid G} $  $ (\Rightarrow\neg\circ_b) $
\end{center}

Using the induction hypothesis, we obtain the following proof in $ \bf FDE^{\bf S5}_{\Box\D}$ (due to space limits we assume that $ H=G=\emptyset $ and omit some names of the rules):
{\small \begin{center}
\EnableBpAbbreviations
\AXC{$ \g(\Theta)\Rightarrow\g(\Lambda),\g(A) $}
\AXC{$ \g(\Theta)\Rightarrow\g(\Lambda),\g(\neg A) $}
\RL{$ (\Rightarrow\wedge) $}
\BIC{$ \g(\Theta)\Rightarrow\g(\Lambda),\g(A^b) $}
\AXC{$ \g(\Gamma)\Rightarrow\g(\Delta),\neg \g(A), \g(A) $}
\UIC{$ \g(\Gamma)\Rightarrow\g(\Delta),\neg \g(A),\neg\neg \g(A) $}
%\LL{$ (\Rightarrow\neg{\wedge}) $}
\UIC{$ \g(\Gamma)\Rightarrow\g(\Delta),\neg\g(A^b) $}
\UIC{$\Rightarrow\neg\Box\g(A^b)\mid\g(\Gamma)\Rightarrow\g(\Delta)$}
\LL{$ (\Rightarrow\neg{\rightarrow}) $}
\BIC{$\g(\Theta)\Rightarrow\g(\Lambda),\neg (\g(A^b)\rightarrow \Box \g(A^b))\mid\g(\Gamma)\Rightarrow\g(\Delta)$}
\RL{Df $ \g $}
\UIC{$\g(\Theta)\Rightarrow\g(\Lambda),\g(\neg\circ_bA)\mid\g(\Gamma)\Rightarrow\g(\Delta)$}
\DisplayProof
\end{center}}

\item The case $ (\neg\circ_b\Rightarrow) $. The last inference of $ \mathfrak{P}_1 $ has the following form:

\begin{center}
 $ \dfrac{A\Rightarrow \;\mid A,\neg A,\Gamma\Rightarrow\Delta\mid H\quad\neg A\Rightarrow\;\mid A,\neg A,\Gamma\Rightarrow\Delta\mid G}{\neg{\circ_b}A,\Gamma\Rightarrow\Delta\mid H\mid G} $  $ (\neg\circ_b\Rightarrow) $
\end{center}

Using the induction hypothesis, we obtain the following proof in $ \bf FDE^{\bf S5}_{\Box\D}$ (due to space limits we assume that $ H=G=\emptyset $ and omit some names of the rules):
{\small \begin{center}
\EnableBpAbbreviations
\AXC{$\g(A)\Rightarrow\mid\g(A),\g(\neg A),\g(\Gamma)\Rightarrow\g(\Delta) $}
%\RL{$ (\non{\Rightarrow}) $}
\UIC{$\neg\neg\g(A)\Rightarrow\mid\g(A),\g(\neg A),\g(\Gamma)\Rightarrow\g(\Delta)$}
%\RL{$ (\wedge{\Rightarrow}) $}

\AXC{$\g(\neg A)\Rightarrow\mid\g(A),\g(\neg A),\g(\Gamma)\Rightarrow\g(\Delta)$}
%\LL{$ (\non{\Rightarrow}) $}
\UIC{$\neg\g(A)\Rightarrow\mid\g(A),\g(\neg A),\g(\Gamma)\Rightarrow\g(\Delta)$}
%\LL{$ (\wedge{\Rightarrow}) $}
\BIC{$\neg\g(A^b)\Rightarrow\mid\g(A),\g(\neg A),\g(\Gamma)\Rightarrow\g(\Delta)$}
\RL{$ (\neg\Box \Rightarrow) $}
\UIC{$\neg\Box\g(A^b)\Rightarrow\mid\g(A),\g(\neg A),\g(\Gamma)\Rightarrow\g(\Delta)$}
\RL{$ (\wedge\Rightarrow) $}
\UIC{$\neg\Box\g(A^b)\Rightarrow\mid\g(A^b),\g(\Gamma)\Rightarrow\g(\Delta)$}
\RL{(Merge)}
\UIC{$\g(A^b),\neg\Box\g(A^b),\g(\Gamma)\Rightarrow\g(\Delta)$}
\RL{$ (\Rightarrow\neg{\rightarrow}) $}
\UIC{$\neg(\g(A^b)\rightarrow\Box\g(A^b)),\g(\Gamma)\Rightarrow\g(\Delta)$}
\RL{Df $ \g $}
\UIC{$\g(\neg{\circ_b}A),\g(\Gamma)\Rightarrow\g(\Delta)$}
\DisplayProof
\end{center}}
\end{enumerate}

The other cases are treated similarly.

(2) By induction on the proof $ \mathfrak{P}_2 $ of $\g(\mathscr{H} )$ in ${\bf FDE^{\bf S5}_{\Box\D}}\setminus \mathrm{(Cut)}$. We distinguish the cases according to the last inference of $ \mathfrak{P}_2 $.

1. The case $ (\rightarrow\Rightarrow) $. The last inference of $ \mathfrak{P}_2 $ is an application of $(\rightarrow\Rightarrow) $.

Subcase 1.1.

\begin{center}
\EnableBpAbbreviations
\AXC{$\g(\Gamma)\Rightarrow\g(\Delta),\g(A)\mid\g(H) $}
\AXC{$\g(B),\g(\Theta)\Rightarrow\g(\Lambda)\mid\g(G) $}
\RL{$ (\rightarrow \Rightarrow) $}
\BIC{$\g(A)\rightarrow\g(B),\g(\Gamma),\g(\Theta)\Rightarrow\g(\Delta),\g(\Lambda)\mid\g(H)\mid\g(G)$}
\DisplayProof
\end{center}

Using the induction hypothesis, we obtain the following proof in $ {\bf FDE^{\bf S5}_{\circ\bullet}} $:
\begin{center}
\EnableBpAbbreviations
\AXC{$ \Gamma\Rightarrow \Delta,A\mid H $}
\AXC{$B,\Theta\Rightarrow \Lambda\mid G $}
\RL{$ (\rightarrow \Rightarrow) $}
\BIC{$ A\rightarrow B,\Gamma,\Theta\Rightarrow \Delta,\Lambda\mid H\mid G$}
\DisplayProof
\end{center}

Subcase 1.2. Recall that $\g(\circ_TA)=\g(A)\rightarrow \Box\g(A) $.

\begin{center}
\EnableBpAbbreviations
\AXC{$\g(A),\g(\Gamma)\Rightarrow\g(\Delta)\mid\g(H) $}
\RL{$ (\Box \Rightarrow) $}
\UIC{$\Box\g(A)\Rightarrow\; \mid\g(\Gamma)\Rightarrow\g(\Delta)\mid\g(H) $}
\AXC{$\g(\Theta)\Rightarrow\g(\Lambda),\g(A)\mid\g(G) $}
\RL{$ (\rightarrow \Rightarrow) $}
\BIC{$\g(A)\rightarrow \Box\g(A),\g(\Theta)\Rightarrow\g(\Lambda)\mid\g(\Gamma)\Rightarrow\g(\Delta)\mid\g(H)\mid\g(G)$}
\DisplayProof
\end{center}

Using the induction hypothesis, we obtain the following proof in $ {\bf FDE^{\bf S5}_{\circ\bullet}} $:
\begin{center}
\EnableBpAbbreviations
\AXC{$ A,\Gamma\Rightarrow \Delta\mid H$}
\AXC{$\Theta\Rightarrow \Lambda, A\mid G $}
\RL{$ (\circ_T \Rightarrow) $}
\BIC{$ \circ_TA,\Theta\Rightarrow \Lambda\mid \Gamma\Rightarrow \Delta\mid H\mid G$}
\DisplayProof
\end{center}

The other cases are treated similarly.
\end{proof}

\begin{theorem}[Syntactical embedding from $\bf FDE^{\bf S5}_{\circ\bullet}$ into $\bf FDE^{\bf S5}_{\Box\D}$]\label{Synt4to4}
For each  hypersequent $ \mathscr{H} $, it holds that 
\begin{enumerate}[$ (1) $]\itemsep=0pt
    \item ${\bf FDE^{\bf S5}_{\circ\bullet}}\vdash \mathscr{H}$ iff ${\bf  FDE^{\bf S5}_{\Box\D}}\vdash\g(\mathscr{H})$,
    \item ${\bf FDE^{\bf S5}_{\circ\bullet}}\setminus \mathrm{(Cut)}\vdash \mathscr{H}$ iff ${\bf  FDE^{\bf S5}_{\Box\D}}\setminus \mathrm{(Cut)}\vdash\g(\mathscr{H})$.
\end{enumerate}
\end{theorem}
\begin{proof}
(1) By Lemma \ref{WSynt4to4}, ${\bf FDE^{\bf S5}_{\circ\bullet}}\vdash \mathscr{H}$ implies $ {\bf  FDE^{\bf S5}_{\Box\D}}\vdash \g(\mathscr{H})$. For another direction suppose that $ {\bf  FDE^{\bf S5}_{\Box\D}}\vdash \g(\mathscr{H})$. By Theorem \ref{CutAdm1}, $ {\bf  FDE^{\bf S5}_{\Box\D}}\setminus(\mathrm{Cut})\vdash \g(\mathscr{H})$. By Lemma \ref{WSynt4to4}, $ {\bf  FDE^{\bf S5}_{\circ\bullet}}\setminus(\mathrm{Cut})\vdash \mathscr{H}$. Hence, $ {\bf  FDE^{\bf S5}_{\circ\bullet}}\vdash \mathscr{H}$.
				
(2) The proof is similar.
\end{proof}
\begin{theorem}[Cut admissibility]
The rule (Cut) is admissible in $\bf FDE^{\bf S5}_{\circ\bullet}\setminus \mathrm{(Cut)}$.
\end{theorem}
\begin{proof}
Follows from Theorems \ref{Synt4to4} and \ref{CutAdm1}.
\end{proof}
\begin{theorem}[Semantic embedding from $\bf FDE^{\bf S5}_{\circ\bullet}$ into $\bf FDE^{\bf S5}_{\Box\D}$]\label{Sem4to4} 
For every $\Lang_{\circ\bullet}^\neg$-formula $ A $, it holds that ${\bf FDE^{\bf S5}_{\circ\bullet}}\models A$ iff ${\bf FDE^{\bf S5}_{\Box\D}}\models\g(A)$.
\end{theorem}
\begin{proof}
This theorem follows from the following Lemmas \ref{Sem4to4L1} and \ref{Sem4to4L2}.
\end{proof}

\begin{lemma}\label{Sem4to4L1}
For any $ {\bf FDE^{\bf S5}_{\circ\bullet}} $-model $ \mathfrak{M}_1=\langle W,\var\rangle $ we can construct an $ {\bf FDE^{\bf S5}_{\B\D}} $-model $ \mathfrak{M}_2=\langle W,\var^\prime\rangle $ such that for each $\Lang_{\circ\bullet}^\neg$-formula $ \A $ and any $ x\in W$, it is true that 
\begin{enumerate}[$\blacktriangleright$]\itemsep=0pt
\item 
$ \one\in\var(A,x) $ in $ \mathfrak{M}_1 $ iff $ \one\in\var^\prime(\g(A),x) $ in $ \mathfrak{M}_2 $,
\item 
$ \one\in\var(\neg A,x) $ in $ \mathfrak{M}_1 $ iff $ \one\in\var^\prime(\g(\neg A),x) $ in $ \mathfrak{M}_2 $.
\end{enumerate}
\end{lemma}
\begin{proof}
Assume that $ \mathfrak{M}_1=\langle W,\var\rangle $ is $ {\bf FDE^{\bf S5}_{\circ\bullet}} $-model and $ \mathfrak{M}_2=\langle W,\var^\prime\rangle $ is $ {\bf FDE^{\bf S5}_{\B\D}} $-model, where $ \var^\prime $ is a valuation such that for all $ p \in \Prop$ and $ x\in W$:	
\begin{enumerate}[$\blacktriangleright$]\itemsep=0pt
\item $ \one\in\var(p,x) $ in $ \mathfrak{M}_1 $ iff $ \one\in\var^\prime(p,x) $ in $ \mathfrak{M}_2 $,
\item $ \one\in\var(\neg p,x) $ in $ \mathfrak{M}_1 $ iff $ \one\in\var^\prime(\neg p,x) $ in $ \mathfrak{M}_2 $.
	\end{enumerate}
	
The proof is by induction on the construction of the $\Lang_{\circ\bullet}^\neg$-formula $ \A $.
	
	\begin{enumerate}[$ (1) $]\itemsep=0pt
\item  $ \A $ is $ p $ or $ \neg p $. Follows from the assumption and Definition \ref{g}.	

\item $ \A $ is $ B\wedge C $.  By Definition \ref{KripkeBelpanEA}, $ \one\in\var(B\wedge C,x) $ in $ \mathfrak{M}_1 $ iff $\one\in\var(B,x)$ in $ \mathfrak{M}_1 $ and $\one\in\var(C,x)$ in $ \mathfrak{M}_1 $ iff (ind.\;hyp.) $\one\in\var^\prime(\g(B),x)$ in $ \mathfrak{M}_2 $ and $\one\in\var^\prime(\g(C),x)$ in $ \mathfrak{M}_2 $ iff (Def. \ref{Four-valuedS5}) $\one\in\var^\prime(\g(B)\wedge\g(C),x)$ in $ \mathfrak{M}_2 $ iff (Def. \ref{g}) $\one\in\var^\prime(\g(B\wedge C),x)$ in $ \mathfrak{M}_2 $.

\item $ \A $ is $ \neg(B\wedge C) $.  $ \one\in\var(\neg(B\wedge C),x) $ in $ \mathfrak{M}_1 $ iff $ \zero\in\var(B\wedge C,x) $ in $ \mathfrak{M}_1 $ iff (Def. \ref{KripkeBelpanEA}) $\zero\in\var(B,x)$ in $ \mathfrak{M}_1 $ or $\zero\in\var(C,x)$ in $ \mathfrak{M}_1 $ iff (ind.\;hyp.) $\zero\in\var^\prime(\g(B),x)$ in $ \mathfrak{M}_2 $ or $\zero\in\var^\prime(\g(C),x)$ in $ \mathfrak{M}_2 $ iff (Def. \ref{Four-valuedS5}) $\zero\in\var^\prime(\g(B)\wedge\g(C),x)$ in $ \mathfrak{M}_2 $ iff (Def. \ref{g}) $\zero\in\var^\prime(\g(B\wedge C),x)$ in $ \mathfrak{M}_2 $ iff $\one\in\var^\prime(\g(\neg(B\wedge C)),x)$ in $ \mathfrak{M}_2 $.
		
\item $ \A $ is $ \circ_IB $, where $ I\in\{T,F \} $. By Def. \ref{KripkeBelpanEA}, $ \one\in\var({\circ_I} B,x) $ in $ \mathfrak{M}_1 $ iff $I\in\var(B,x)$ in $ \mathfrak{M}_1 $ implies $ \forall_{y\in W}I\in\var(B,y)$ in $ \mathfrak{M}_1 $ 
iff (Lemma \ref{X}) $T\in\var(B^I,x)$ in $ \mathfrak{M}_1 $ implies $ \forall_{y\in W}T\in\var(B^I,y)$ in $ \mathfrak{M}_1 $
iff (by the ind.\;hyp.) $\one\in\var^\prime(\g(B^I),x)$ in $ \mathfrak{M}_2 $ implies $ \forall_{y\in W}\one\in\var^\prime(\g(B^I),y)$ in $ \mathfrak{M}_2 $
iff (by Def. \ref{Four-valuedS5}) $\one\in\var^\prime(\g(B^I),x)$ in $ \mathfrak{M}_2 $ implies $ \one\in\var^\prime(\Box \g(B^I),x)$ in $ \mathfrak{M}_2 $ iff (by Def. \ref{Four-valuedS5}) $\one\in\var^\prime(\g(B^I)\rightarrow\Box \g(B^I),x)$ in $ \mathfrak{M}_2 $ iff (by Def. \ref{g}) $ \one\in\var^\prime(\g({\circ_I} B),x) $ in $ \mathfrak{M}_2 $.

\item $ \A $ is $ \neg{\circ_I}B $, where $ I\in\{T,F \} $. By Def. \ref{KripkeBelpanEA}, $ \one\in\var(\neg{\circ_I} B,x) $ in $ \mathfrak{M}_1 $ iff $ \zero\in\var({\circ_I} B,x) $ in $ \mathfrak{M}_1 $ iff  $I\in\var(B,x)$ in $ \mathfrak{M}_1 $ and $ \exists_{y\in W}\non I\in\var(B,y)$ in $ \mathfrak{M}_1 $ iff (Lemma \ref{X}) $T\in\var(B^I,x)$ in $ \mathfrak{M}_1 $ and $ \exists_{y\in W} T\in\var(B^{\non I},y)$ in $ \mathfrak{M}_1 $ iff $T\in\var(B^I,x)$ in $ \mathfrak{M}_1 $ and $ \exists_{y\in W} F\in\var(B^{I},y)$ in $ \mathfrak{M}_1 $
iff (by the ind.\;hyp.) $\one\in\var^\prime(\g(B^I),x)$ in $ \mathfrak{M}_2 $ and $ \exists_{y\in W}\zero\in\var^\prime(\g(B^{I}),y)$ in $ \mathfrak{M}_2 $ 
iff (Def. \ref{Four-valuedS5}) $\one\in\var^\prime(\g(B^I),x)$ in $ \mathfrak{M}_2 $ and $ \zero\in\var^\prime(\Box\g(B^{I}),x)$ in $ \mathfrak{M}_2 $ iff $\zero\in\var^\prime(\g(B^I)\rightarrow\Box\g(B^I),x)$ in $ \mathfrak{M}_2 $ iff $\one\in\var^\prime(\neg(\g(B^I)\rightarrow\Box\g(B^I)),x)$ in $ \mathfrak{M}_2 $ iff (Def. \ref{g}) 
 $\zero\in\var^\prime(\g({\circ_I}B),x)$ in~$ \mathfrak{M}_2 $ iff $\one\in\var^\prime(\g(\neg{\circ_I}B),x)$ in~$ \mathfrak{M}_2 $. 

\item $ A $ is $ \circ_iB $, where $ i\in\{t,b,n,f \} $.  By Def. \ref{KripkeBelpanEA}, $ \one\in\var({\circ_i} B,x) $ in $ \mathfrak{M}_1 $ iff  
 $\var(B,x)=i$ in $ \mathfrak{M}_1 $ implies $ \forall_{y\in W} \var(B,y)=i$ in $ \mathfrak{M}_1 $ iff (Lemma \ref{X}) 
 $T\in\var(B^i,x)$ in $ \mathfrak{M}_1 $ implies $ \forall_{y\in W} T\in\var(B^i,y)$ in $ \mathfrak{M}_1 $ 
  iff (this lemma, (2)(a))
 $\one\in\var^\prime(\g(B^i),x)$ in $ \mathfrak{M}_2 $ implies $ \forall_{y\in W}\one\in \var^\prime(\g(B^i),y)$ in $ \mathfrak{M}_2 $  
  iff (Def. \ref{Four-valuedS5})  
  $\one\in\var^\prime(\g(B)^i),x)$ in $ \mathfrak{M}_2 $ implies $ \one\in \var^\prime(\Box(\g(B^i)),x)$ in $ \mathfrak{M}_2 $  
  iff (Def. \ref{Four-valuedS5})
   $\one\in\var^\prime((\g(B^i)\rightarrow\Box(\g(B^i)),x)$ in $ \mathfrak{M}_2 $ 
  iff (Def. \ref{g}) $ \one\in\var^\prime(\g({\circ_i} B),x) $~in~$ \mathfrak{M}_2 $. 
  
\item $ A $ is $ \neg{\circ_i}B $, where $ i\in\{t,b,n,f \} $.  By Def. \ref{KripkeBelpanEA}, $ \one\in\var(\neg{\circ_i} B,x) $ in $ \mathfrak{M}_1 $ iff  $ \zero\in\var({\circ_i} B,x) $ in $ \mathfrak{M}_1 $ iff 
$\var(B,x)=i$ in $ \mathfrak{M}_1 $ and $ \exists_{y\in W}\var(B,y)\not=\ast i$ in $ \mathfrak{M}_1 $
 iff (Lemma \ref{X})  
$T\in\var(B^i,x)$ in $ \mathfrak{M}_1 $ and $ \exists_{y\in W}F\in\var(B^i,y)$ in $ \mathfrak{M}_1 $ 
 iff 
(this lemma, (2))
$\one\in\var^\prime(\g(B^i),x)$ in $ \mathfrak{M}_2 $ and $ \exists_{y\in W}\zero\in\var^\prime(\g(B^i),y)$ in $ \mathfrak{M}_2 $ 
 iff (Def. \ref{Four-valuedS5})
$\one\in\var^\prime(\g(B^i),x)$ in $ \mathfrak{M}_2 $ and $ \zero\in\var^\prime(\Box\g(B^i),x)$ in $ \mathfrak{M}_2 $ iff (Def. \ref{Four-valuedS5})   
    $\zero\in\var^\prime(\g(B^i)\rightarrow\B \g(B^i),x)$  in $ \mathfrak{M}_2 $  
    iff (Def. \ref{g})        
  $ \zero\in\var^\prime(g({\circ_i} B),x) $  in $ \mathfrak{M}_2 $ iff $ \one\in\var^\prime(g(\neg{\circ_i} B),x) $  in $ \mathfrak{M}_2 $.
	\end{enumerate}
	
	The other cases are treated similarly.
\end{proof} 
\begin{lemma}\label{Sem4to4L2}
For any $ {\bf FDE^{\bf S5}_{\B\D}} $-model $ \mathfrak{M}_1=\langle W,\var\rangle $ we can construct an $ {\bf FDE^{\bf S5}_{\circ\bullet}} $-model $ \mathfrak{M}_2=\langle W,\var^\prime\rangle $ such that for each $\Lang_{\circ\bullet}^\neg$-formula $ \A $ and any $ x\in W$, it is true that 
\begin{enumerate}[$\blacktriangleright$]\itemsep=0pt
\item 
	$ \one\in\var(\g(A),x) $ in $ \mathfrak{M}_1 $ iff $ \one\in\var^\prime(A,x) $ in $ \mathfrak{M}_2 $,
	\item 
	$ \one\in\var(\g(\neg A),x) $ in $ \mathfrak{M}_1 $ iff $ \one\in\var^\prime(\neg A,x)$ in $ \mathfrak{M}_2 $.
	\end{enumerate}
\end{lemma}
\begin{proof}
	Similarly to Lemma \ref{Sem4to4L1}.
\end{proof}
\begin{theorem}[Soundness and completeness of $ {\bf FDE^{\bf S5}_{\circ\bullet}} $]\label{FDES5EASoundCompl}
For every $\Lang_{\circ\bullet}^\neg$-formula $ A $, it holds that ${\bf FDE^{\bf S5}_{\circ\bullet}}\models A$ iff ${\bf FDE^{\bf S5}_{\circ\bullet}}\vdash A$.
\end{theorem}
\begin{proof}
Follows from Theorems \ref{Synt4to4} and \ref{Sem4to4}.
\end{proof}
\begin{theorem}[Decidability]
$ {\bf FDE^{\bf S5}_{\B\D}} $ and $ {\bf FDE^{\bf S5}_{\circ\bullet}} $ are decidable.
\end{theorem}
\begin{proof}
Follows from the decidability of \textbf{S5}, Theorems \ref{Synt4to2} and  \ref{Synt4to4}.
\end{proof}
\section{Conclusion}\label{Conclusion} 
 We presented many-valued versions of essence and accident modalities, some of them (e.g., $ \circ_b $ or $ \circ_n $) cannot be expressed in two-valued modal logic. Let us mention some potential avenues for further investigation. 
It is well-known that by excluding $ b $ from the set of \textbf{FDE}'s truth values one gets strong Kleene $ \bf K_3$ \cite{Kleene38}. Similarly, by eliminating $ n $ one gets Asenjo--Priest's \cite{Asenjo66,PriestLP} logic of paradox \textbf{LP}. Thus,  it is feasible to conduct a more comprehensive examination of the three-valued versions of essence and accident modalities.  Alternatively, one may consider a broader approach and explore these modalities in a $ n $-valued logic, $ n>2 $. A possible approach to defining essence and accident is by using Rosser and Turquette's $ J $-operators \cite{RosserTurquette}. The following definition may be proposed, where
$ k\in\{1,\frac{n-2}{n-1},\ldots,\frac{1}{n-1},0 \} $:
\begin{enumerate}[$ - $]\itemsep=0pt
\vspace{-.1in}{\small \begin{multicols}{2}
\item $ \circ_kA=J_kA\rightarrow\Box J_kA $,
\item $ \bullet_kA=J_kA\wedge\D\non J_kA $,
\item $ \FE_kA=\non J_kA\rightarrow\Box\non J_kA $,
\item $ \FA_kA=\non J_kA\wedge\D J_kA $.
\end{multicols}}\vspace{-.15in}
\end{enumerate}

One may also consider essence and accident modalities in the context of Zaitsev and Shramko's bi-facial logic of generalized classical truth values \cite{ZS}, which differentiates ontological and epistemic truth and falsity and has the following four truth values: $ \bf T1$, $ \bf T0$, $ \bf F1$, and $ \bf F0$, where $ \bf T$ and $ \bf F$ are ontological components of truth values and $ \bf 1$ and $ \bf 0$ are epistemic ones. Then one can say, e.g., that $ A $ is essentially ontologically true, or $ A $ is accidentally epistemically false.

Clearly, one may contemplate alternative modal logics beyond \textbf{S5}. A recent paper \cite{PetrukhinKBTD4} introduced cut-free nested sequent calculi for essence and accident logics based on \textbf{K}. One may also consider alternative normal modal logics, e.g., \textbf{K5}, or \textbf{GL}, or some regular, monotonic, or congruent modal logics.

Another topic is the development of a first-order version of the hypersequent calculi in question (see \cite{AvronLahav} for a calculus similar to Restall's one \cite{Restall} for a first-order version of \textbf{S5} with a paraconsistent negation). To enhance this calculus, one can incorporate definite descriptions. There are two approaches to consider: the traditional one using the $ \iota $-term forming operator, or the recent approach by Wa\l{}ega and Zawidzki \cite{WalegaZawidzki}, which is a continuation of the research started in \cite{IndrzejczakZawidzki} and is based on hybrid logic and uses an operator $ @_{\varphi}\psi $ to indicate that `$ \psi $ holds in the modal world where $ \varphi $ holds'. One could introduce a many-valued modification of this operator: $ @^i_{\varphi}\psi$ indicating that `$ \psi $ takes the truth value $ i $ in the modal world where $ \varphi $ takes the value $ i $'. Then one could add the following operators: $ @^{\circ i}_{\varphi}\psi$ indicating that `$ \psi $ essentially takes the value $ i $ in the modal world where $ \varphi $ essentially takes the value $ i $' and $ @^{\bullet i}_{\varphi}\psi$ indicating that `$ \psi $ accidentally takes the value $ i $ in the modal world where $ \varphi $ accidentally takes the value $ i $'. It is also possible to introduce a kind of combination of these operators: $ @^{\bullet\circ i}_{\varphi}\psi$ indicating that `$ \psi $ \textit{accidentally} takes the value $ i $ in the modal world where $ \varphi $ \textit{essentially} takes the value $ i $'. However, we believe that this topic requires a separate paper to be properly investigated.

\paragraph{Acknowledgments.} The author thanks the reviewers for their useful comments. Funded by the European Union (ERC, ExtenDD, project number: 101054714). Views and opinions expressed are however those of the author(s) only and do not necessarily reflect those of the European Union or the European Research Council. Neither the European Union nor the granting authority can be held responsible for them.


\begin{thebibliography}{99}

\bibitem{ArieliAvron}
Arieli, O. and Avron, A., ``Reasoning with Logical Bilattices,'' \emph{Journal of Logic, Language and Information}, 5(1) (1996), 25--63.

\bibitem{AvronLahav}
Avron, A., and Lahav, O., ``A Simple Cut-Free System for a Paraconsistent Logic Equivalent to S5,'' in \emph{Advances in modal logic, Volume 12}, College Publications, 2018, 29--42.

\bibitem{Drobyshevich}
Drobyshevich, S., ``A General Framework for FDE-Based Modal Logics,'' \emph{Studia Logica}, 108 (2020), 1281--1306.

\bibitem{Fitting92}
Fitting, M., ``Many-valued modal logics II,'' \emph{Fundamenta Informaticae}, 17 (1992), 55--73.

\bibitem{Goble}
Goble, L., ``Paraconsistent modal logic,'' \emph{Logique et Analyse}, 49 (2006), 3--29.

\bibitem{OdintsovSkurtWansing}
Odintsov, S.P. and Skurt, D. and Wansing, H., ``On Definability of Connectives and Modal Logics over FDE,'' \emph{Logic and Logical Philosophy}, 28 (2019), 631--659.

\bibitem{OdintsovLatkin}
Odintsov, S.P. and Latkin, E.I., ``BK-lattices. Algebraic semantics for Belnapian modal logics,'' \emph{Studia Logica}, 100(1--2) (2012), 319--338.

\bibitem{OdintsovSperanski}
Odintsov, S.P. and Speranski, S., ``The lattice of Belnapian modal logics: special extensions and counterparts,'' \emph{Logic and Logical Philosophy}, 25 (2016), 3--33.

\bibitem{OdintsovSperanski19}
Odintsov, S.P. and Speranski, S., ``Belnap--Dunn modal logics: truth constants vs. truth values,'' \emph{The Review of Symbolic Logic}, 13(2) (2020), 416--435.

\bibitem{Rivieccio4}
Rivieccio, U. and Jung, A. and Jansana, R., ``Four-valued modal logic: Kripke semantics and duality,'' \emph{Journal of Logic and Computation}, 27 (2017), 155--199.

\bibitem{Sedlar}
Sedl\'{a}r, I., ``Propositional dynamic logic with Belnapian truth values,'' in \emph{Advances in Modal Logic}, vol. 11, 2016, 503--519.

\bibitem{KozhVashentsevaNonCont}
Kozhemiachenko, D. and Vashentseva, L., ``Non-contingency in a paraconsistent setting,'' \emph{Logic Journal of the IGPL}, 32(3) (2024), 394--419.

\bibitem{KozhVashentsevaIgn}
Kozhemiachenko, D. and Vashentseva, L., ``Knowledge and ignorance in Belnap--Dunn logic,'' \emph{Logic Journal of the IGPL}, 33(1) (2025), 125--148.

\bibitem{Sherkhonov2008}
Sherkhonov, E. Yu., ``Modal operators over constructive logic,'' \emph{Journal of Logic and Computation}, 18(6) (2008), 815--829.

\bibitem{Avron}
Avron, A., ``Natural 3-valued logics~--- characterization and proof theory,'' \emph{Journal of Symbolic Logic}, 56(1) (1991), 276--294.

\bibitem{AvronClassical}
Avron, A., ``Classical Gentzen-type methods in propositional many-valued logics,'' in \emph{Beyond Two: Theory and Applications of Multiple-Valued Logic}, Physica, 2003, 117--155.

\bibitem{BednarskaIndrzejczak}
Bednarska, K. and Indrzejczak, A., ``Hypersequent calculi for {S}5: the methods of cut elimination,'' \emph{Logic and Logical Philosophy}, 24(3) (2015), 277--311. %doi: \url{https://doi.org/10.12775/LLP.2015.018}.

\bibitem{Belnap1}
Belnap, N.D., ``A useful four-valued logic,'' in \emph{Modern Uses of Multiple-Valued Logic}, Reidel Publishing Company, 1977, 7--37.

\bibitem{Belnap2}
Belnap, N.D., ``How a computer should think,'' in \emph{Contemporary Aspects of Philosophy}, Oriel Press, 1977, 30--56.

\bibitem{DeOmori}
De, M. and Omori, H., ``Classical Negation and Expansions of {B}elnap-{D}unn Logic,'' \emph{Studia Logica}, 103(4) (2015), 825--851.

\bibitem{Dunn76}
Dunn, J.M., ``Intuitive semantics for first-degree entailment and coupled trees,'' \emph{Philosophical Studies}, 29 (1976), 149--168.

\bibitem{Fan}
Fan, J., ``Logics of Essence and Accident,'' \emph{arXiv:1506.01872v1} (2015).

\bibitem{Fine1}
Fine, K., ``Essence and modality,'' \emph{Philosophical Perspectives}, 8 (1994), 1--16.

\bibitem{Fine2}
Fine, K., ``The logic of essence,'' \emph{Journal of Philosophical Logic}, 24 (1995), 241--273.

\bibitem{Fine3}
Fine, K., ``Semantics for the logic of essence,'' \emph{Journal of Philosophical Logic}, 29 (2000), 543--584.

\bibitem{Fitting91}
Fitting, M., ``Many-valued modal logics,'' \emph{Fundamenta Informaticae}, 15 (1991), 235--254.

\bibitem{GilbertVenturi}
Gilbert, D.\ts R. and Venturi, G., ``A note on logics of essence and accident,'' \emph{Logic Journal of the IGPL}, 28(5) (2020), 881--891.

\bibitem{MMLS5}
Grigoriev, O. and Petrukhin, Y., ``On a multilattice analogue of a hypersequent {S}5 calculus,'' \emph{Logic and Logical Philosophy}, 28(4) (2019), 683--730. %doi: \url{https://doi.org/10.12775/ LLP.2019.031}.

\bibitem{Gurevich}
Gurevich, Y., ``Intuitionistic logic with strong negation,'' \emph{Studia Logica}, 36 (1977), 49--59.

\bibitem{Indrzejczak2019}
Indrzejczak, A., ``Two Is Enough --- Bisequent Calculus for {S}5,'' in \emph{Frontiers of Combining Systems. FroCoS 2019}, Springer, 2019, 277--294.

\bibitem{Indrzejczak2020}
Indrzejczak, A., ``Existence, definedness and definite descriptions in hybrid modal logic,'' in \emph{Advances in modal logic, vol. 13}, College publications, 2020, 349--368.

\bibitem{IndrzejczakND}
Indrzejczak, A., \emph{Natural Deduction, Hybrid Systems and Modal Logics}, Springer, 2010. %doi: \url{https://doi.org/10.1007/978-90-481-8785-0}.

\bibitem{MarcosPara}
Marcos, J., ``Nearly Every Normal Modal Logic is Paranormal,'' \emph{Logique et Analyse}, 48(189--192) (2005), 279--300.

\bibitem{Boolos}
Boolos, G., \emph{The Logic of Provability}, Cambridge University Press, 1993. %doi: %\url{https://doi.org/10.1017/CBO9780511625183}.

\bibitem{Steinsvoldboxdot}
Steinsvold, C., ``The boxdot conjecture and the language of essence and accident,'' \emph{Australasian Journal of Logic}, 10 (2011), 18--35.

\bibitem{Steinsvold}
Steinsvold, C., ``Completeness for Various Logics of Essence and Accident,'' \emph{Bulletin of the Section of Logic}, 37(2) (2008), 93--101.

\bibitem{IndrzejczakSC}
Indrzejczak, A., \emph{Sequents and Trees: An Introduction to the Theory and Applications of Propositional Sequent Calculi}, Birkhauser Cham, 2021.

\bibitem{IndrzejczakZawidzki}
Indrzejczak, A. and Zawidzki, M., ``Definite descriptions and hybrid tense logic,'' \emph{Synthese}, 202 (2023), 98.

\bibitem{Kamide19}
Kamide, N., ``Gentzen-type sequent calculi for extended {B}elnap--{D}unn logics with classical negation: A general framework,'' \emph{Logica Universalis}, 13(1) (2019), 37--63.

\bibitem{Kamide23}
Kamide, N., ``Falsification-Aware Calculi and Semantics for Normal Modal Logics Including {S}4 and {S}5,'' \emph{Journal of Logic, Language and Information}, 32(3) (2023), 395--440.

\bibitem{Asenjo66}
Asenjo, F.G., ``A calculus of antinomies,'' \emph{Notre Dame Journal of Formal Logic}, 7 (1966), 103--105.

\bibitem{Kleene38}
Kleene, S.C., ``On a notation for ordinal numbers,'' \emph{The Journal of Symbolic Logic}, 3 (1938), 150--155.

\bibitem{PriestLP}
Priest, G., ``The logic of paradox,'' \emph{Journal of Philosophical Logic}, 8 (1979), 219--241.

\bibitem{MML}
Kamide, N. and Shramko, Y., ``Modal Multilattice Logic,'' \emph{Logica Universalis}, 11(3) (2017), 317--343.

\bibitem{Marcos}
Marcos, J., ``Logics of essence and accident,'' \emph{Bulletin of the Section of Logic}, 34(1) (2005), 43--56.

\bibitem{OdintsovWansing}
Odintsov, S.P. and Wansing, H., ``Modal logics with Belnapian truth values,'' \emph{Journal of Applied Non-Classical Logics}, 20 (2010), 279--301.

\bibitem{OdintsovWansing17}
Odintsov, S.P. and Wansing, H., ``Disentangling {F}{D}{E}-based paraconsistent modal logics,'' \emph{Studia Logica}, 105 (2017), 1221--1254.

\bibitem{OmoriWansing17}
Omori, H. and Wansing, H., ``40 years of {F}{D}{E}: an introductory overview,'' \emph{Studia Logica}, 105 (2017), 1021--1049.

\bibitem{PetrukhinShanginFDE}
Petrukhin, Y. and Shangin, V., ``Correspondence analysis and automated proof-searching for first degree entailment,'' \emph{European Journal of Mathematics}, 6(4) (2020), 1452--1495.

\bibitem{PetrukhinS5}
Petrukhin, Y., ``{S}5-Style Non-Standard Modalities in a Hypersequent Framework,'' \emph{Logic and Logical Philosophy}, 31(3) (2022), 427--456.

\bibitem{PetrukhinKBTD4}
Petrukhin, Y., ``Nested Sequent Calculi for Some Modal Logics with Non-Standard Modalities,'' \emph{Logic and Logical Philosophy}, 34(2) (2025), 287--318.

\bibitem{PetrukhinFDE}
Petrukhin, Y., ``Correspondence analysis for first degree entailment,'' \emph{Logical Investigations}, 22(1) (2016), 108--124.

\bibitem{Poggiolesi}
Poggiolesi, F., \emph{Gentzen calculi for modal propositional logic}, Springer, 2011.

\bibitem{Priest08b}
Priest, G., ``Many-valued modal logics: a simple approach,'' \emph{The Review of Symbolic Logic}, 1(2) (2008), 190--203.

\bibitem{Popov}
Popov, V.M., ``Sequent formulations of paraconsistent logical systems (in {R}ussian),'' in \emph{Semantic and syntactic investigations of non-extensional logics}, Nauka Publ., 1989, 285--289.

\bibitem{Rautenberg}
Rautenberg, W., ``Klassische und Nicht-klassische Aussagenlogik,'' \emph{Vieweg, Braunschweig} (1979).

\bibitem{Restall}
Restall, G., ``Proofnets for {S}5: {S}equents and circuits for modal logic,'' in \emph{Logic Colloquium 2005}, Cambridge University Press, 2007, 151--172.

\bibitem{RosserTurquette}
Rosser, J.B. and Turquette, A.R., \emph{Many-valued logics}, Amsterdam, North-Holland, 1952.

\bibitem{WalegaZawidzki}
Wa\l{}ega, P.A. and Zawidzki, M., ``Hybrid Modal Operators for Definite Descriptions,'' in \emph{Logics in Artificial Intelligence. JELIA 2023. Lecture Notes in Computer Science, vol. 14281}, Springer, 2023, 712--726.

\bibitem{Voroviev}
Vorob'ev, N.N., ``A constructive propositional calculus with strong negation (in {R}ussian),'' \emph{Dokl. Akad. Nauk SSSR}, 85 (1952), 465--468.

\bibitem{ZS}
Zaitsev, D. and Shramko, Y., ``Bi-facial Truth: a Case for Generalized Truth Values,'' \emph{Studia Logica}, 101 (2013), 1299--1318.

\end{thebibliography}
\end{document}